\documentclass[journal]{vgtc}                     

\vgtccategory{Research}

\title{\texttt{BH-tsNET}, \texttt{FIt-tsNET}, \texttt{L-tsNET}: \\
Fast tsNET Algorithms for Large Graph Drawing}

\author{Amyra Meidiana, Seok-Hee Hong,  and Kwan-Liu Ma}

\authorfooter{
  \item
  	Amyra Meidiana is with the University of Sydney.
  	E-mail: amyra.meidiana@sydney.edu.au
  \item
  	Seok-Hee Hong is with the  University of Sydney.
  	E-mail: seokhee.hong@sydney.edu.au.

  \item Kwan-Liu Ma is with the University of California at Davis.
  	E-mail: ma@cs.ucdavis.edu.
}

\abstract{%
The \texttt{tsNET} algorithm utilizes \texttt{t-SNE} to compute high-quality graph drawings, preserving the neighborhood and clustering structure of vertices. 
In this paper, we present three fast algorithms for reducing the time complexity of \texttt{tsNET} algorithm from $O(nm)$ time to $O(n \log n)$ time and $O(n)$ time.
To reduce the runtime of \texttt{tsNET}, there are three components that need to be reduced: (C0) computation of high-dimensional probabilities, (C1) computation of KL divergence gradient, and (C2) entropy computation. 
Specifically, we reduce the overall runtime of \texttt{tsNET}, integrating our new fast approaches for C0 and C2 with fast \texttt{t-SNE} algorithms for C1.
We first present $O(n \log n)$-time \texttt{BH-tsNET}, based on (C0) new  $O(n)$-time partial BFS-based high-dimensional probability computation and (C2) new $O(n \log n)$-time quadtree-based entropy computation, integrated with (C1) $O(n \log n)$-time quadtree-based KL divergence computation of \texttt{BH-SNE}.
We next present faster $O(n \log n)$-time \texttt{FIt-tsNET}, using (C0) $O(n)$-time partial BFS-based high-dimensional probability computation and (C2) quadtree-based $O(n \log n)$-time entropy computation, integrated with (C1) $O(n)$-time interpolation-based KL divergence computation of \texttt{FIt-SNE}.
Finally, we present the fastest $O(n)$-time \texttt{L-tsNET}, integrating (C2) new $O(n)$-time  FFT-accelerated interpolation-based entropy computation with (C0) $O(n)$-time partial BFS-based high-dimensional probability computation, and (C1) $O(n)$-time interpolation-based KL divergence computation of \texttt{FIt-SNE}.
Extensive experiments using benchmark data sets confirm that \texttt{BH-tsNET}, \texttt{FIt-tsNET}, and \texttt{L-tsNET} outperform \texttt{tsNET}, running 93.5\%, 96\%, and 98.6\% faster while computing similar quality drawings in terms of quality metrics (neighborhood preservation, stress, edge crossing, and shape-based metrics) and visual comparison.
We also present a comparison between our algorithms and DR graph, another dimension reduction-based graph drawing algorithm.
}

\keywords{Large Graph Drawing, \texttt{tsNET} }

\graphicspath{{figs/}{figures/}{pictures/}{images/}{./}} 

\usepackage{tabu}                      
\usepackage{booktabs}                  
\usepackage{lipsum}                    
\usepackage{mwe}                       

\usepackage{mathptmx}                  

\usepackage{amsmath}
\usepackage{amsthm}
\usepackage{amsbsy}
\usepackage{algorithm}
\usepackage{algpseudocode}
\algnewcommand\algorithmicin{\textbf{Input}}

\newtheorem{hyp}{Hypothesis}
\newtheorem{theorem}{Theorem}
\usepackage{comment}
\usepackage{bm}
\mathchardef\mhyphen="2D 

\begin{document}


\firstsection{Introduction}

\maketitle

The \texttt{tsNET}~\cite{kruiger2017graph} algorithm for graph drawing utilizes the popular dimensional reduction (DR) algorithm \texttt{t-SNE} \emph{(t-Stochastic Neighborhood Embedding)}~\cite{van2008visualizing} that specifically aims to preserve the \emph{neighborhood} of data points in a low-dimensional projection, unlike most DR algorithms that focus on preserving similarity/distance. 
\texttt{t-SNE} models the distances of data points in the high- and low-dimensional spaces as probability distributions
and then computes a low-dimensional projection by minimizing the \emph{KL (Kullback-Leibler) divergence}~\cite{kullback1997information} between the two probability distributions,
typically using gradient descent, by moving the points in the projection according to the gradient of the KL divergence.
The runtime of \texttt{t-SNE} is $O(n^2)$, due to the computation of pairwise distances between all pairs of data points in both the high- and low-dimensional spaces.

\texttt{tsNET} uses the {\em graph theoretic distance} (i.e., shortest path) between vertices in a graph $G = (V, E)$ as the high-dimensional distance to be projected to a low-dimensional space, 
and then adds the following two more terms to the cost function, an $O(n)$-time \emph{compression} term to accelerate untangling the graph, and an $abstract$-time \emph{entropy} term for reducing clutter,
on top of the KL divergence term used by  \texttt{t-SNE}.
\texttt{tsNET} is able to compute high-quality graph drawings, preserving the neighborhood and clustering structure of vertices. However, its overall runtime is $O(nm)$ (where $n = |V|$, $m = |E|$), 
i.e., $O(n^3)$ time since $m = O(n^2)$ for dense graphs, therefore not scalable to larger graphs.

Recently, fast \texttt{t-SNE} algorithms have been introduced. For example, \texttt{BH-SNE} \emph{(Barnes-Hut SNE)}~\cite{van2014accelerating} utilizes $k$NN ($k$-Nearest Neighbors) and quadtree to reduce the computation of high-dimensional probabilities and KL divergence gradient to $O(n \log n)$ time.
More recently, \texttt{FIt-SNE} \emph{(FFT-accelerated Interpolation-based t-SNE)}~\cite{linderman2019fast} utilizes $O(n \log n)$-time approximate $k$NN to compute the high-dimensional probabilities, and a FFT (Fast Fourier Transform)~\cite{nussbaumer1981fast}-accelerated interpolation to further improve the runtime of the KL divergence gradient computation to $O(n)$ time.

However, algorithms reducing the overall runtime of \texttt{tsNET} have not been investigated yet. 
To reduce the runtime of \texttt{tsNET}, three components need to be reduced: (C0) computation of high-dimensional probabilities, (C1) computation of KL divergence gradient, and (C2) entropy computation (note that the \emph{compression} term runs in $O(n)$ time). 
While it was mentioned as possible future work to use fast \texttt{t-SNE} methods to improve the runtime of the KL divergence term~\cite{kruiger2017graph}, 
methods to reduce the runtime of the {\em entropy} term were not considered yet.
Therefore, there is a significant gap in the literature to address fast \texttt{tsNET} algorithms.

In this paper, we present three new fast \texttt{tsNET} algorithms for graph drawing, integrating our new fast approaches for C0 and C2 with fast \texttt{t-SNE} methods for C1 
to reduce the overall runtime of \texttt{tsNET}. 
First, \texttt{BH-tsNET} and \texttt{FIt-tsNET} incorporate our new partial BFS approach for C0 and quadtree-based entropy computation for C2 with the fast KL divergence computation of \texttt{BH-SNE} and \texttt{FIt-SNE} respectively for C1 to reduce the overall runtime to $O(n \log n)$.
Our fastest algorithm, \texttt{L-tsNET}, uses our new interpolation-based entropy computation for C2 on top of the partial BFS for C0 and the fast KL divergence computation of \texttt{FIt-SNE} for C1, reducing the overall runtime to $O(n)$.

\begin{table}[h]
    \centering
    \scriptsize
    \caption{Time Complexity of \texttt{BH-tsNET}, \texttt{FIt-tsNET} and \texttt{L-tsNET}  }
    \begin{tabular}{|c|c|c|c|c|}
        \hline 
        Algorithm & C0 & C1 & C2 & Overall runtime \\ \hline
        \texttt{tsNET} & $O(nm)$ & $O(n^2)$ & $O(n^2)$ & $O(nm)$ \\ \hline
            \texttt{BH-tsNET} & $O(n)$ & $O(n \log n)$ & $O(n \log n)$ & $O(n \log n)$ (Thm. \ref{theorem:bhtsnet}) \\ \hline
        \texttt{FIt-tsNET} & $O(n)$ & $O(n)$ & $O(n \log n)$ & $O(n \log n)$ (Thm. \ref{theorem:fittsnet}) \\ \hline
        \texttt{L-tsNET} & $O(n)$  & $O(n)$  & $O(n)$  & $O(n)$ (Thm. \ref{theorem:ltsnet})  \\ \hline
    \end{tabular}
    \label{tab:runtime_tsnet}
\end{table}

Table \ref{tab:runtime_tsnet} shows   
the theoretical time complexity analysis of our fast algorithms against the original \texttt{tsNET} per component.
%
Specifically, the main contribution of this paper can be summarized as follows:

\begin{enumerate}

    \item 
    We present $O(n \log n)$-time \texttt{BH-tsNET}, which uses new $O(n)$-time partial BFS-based high-dimensional probability computation for C0 and quadtree-based $O(n \log n)$-time entropy computation for C2, integrated with $O(n \log n)$-time \texttt{BH-SNE} quadtree-based KL divergence computation of for C1.

     \item 
     We present $O(n \log n)$-time \texttt{FIt-tsNET}, which uses new $O(n)$-time partial BFS-based high-dimensional probability computation for C0,  quadtree-based $O(n \log n)$-time entropy computation for C2, integrated with $O(n)$-time \texttt{FIt-SNE} interpolation-based KL divergence computation for C1.

    \item 
    We finally present the fastest $O(n)$-time \texttt{L-tsNET}, which uses new $O(n)$-time  FFT-accelerated interpolation-based entropy computation for C2 with $O(n)$-time partial BFS-based high-dimensional probability computation for C0, integrated with $O(n)$-time interpolation-based KL divergence computation of \texttt{FIt-SNE} for C1.
    
    \item We implement our algorithms and evaluate against \texttt{tsNET} using a variety of benchmark data sets on runtime, quality metrics, and visual comparison. 
    Extensive experiments confirm that \texttt{BH-tsNET}, \texttt{FIt-tsNET}, and \texttt{L-tsNET} outperform  \texttt{tsNET}, running 93.5\%, 96\%, and 98.6\% faster while computing similar quality drawings in terms of quality metrics (neighborhood preservation, stress, edge crossing, and shape-based metrics) and visual comparison.

    \item We also compare our algorithms to \texttt{DRGraph} ~\cite{zhu2020drgraph}, another linear-time \texttt{t-SNE}-based graph drawing algorithm with a different optimization function.     Experiments show that our algorithms especially perform better over \texttt{DRGraph} in visualizing graphs with a regular structure, such as mesh graphs.
\end{enumerate}

\section{Related Work}
\label{sec:litrev}

\subsection{tsNET Algorithm for Graph Drawing}

The \texttt{tsNET}~\cite{kruiger2017graph} algorithm optimizes a cost function $C$ consisting of the following three terms: 

\begin{equation}
\label{eq:tsnet}
\begin{split}
C = C_{KL} + C_{CMP} + C_{ENT}, \\
C_{KL} = \lambda_{KL} \sum_{i \neq j} p_{ij} \log \frac{p_{ij}}{q_{ij}},\\
C_{CMP} = \frac{\lambda_c}{2n}\sum_i ||X_i||^2, \\
C_{ENT} = - \frac{\lambda_r}{2n^2} \sum_{i \neq j} \log (||X_i - X_j|| + \epsilon)
\end{split}
\end{equation}

The first term $C_{KL}$ is the KL divergence optimized by \texttt{t-SNE} \emph{(t-Stochastic Neighborhood Embedding)}~\cite{van2008visualizing}, 
a DR algorithm that aims to preserve the neighborhood of data points by modeling the pairwise similarities of the objects in high-dimensional space and low-dimensional projection as probability distributions, where objects that are closer or more similar have a higher probability.
The low-dimensional projection is then computed by minimizing the divergence between the two distributions. 

Specifically, for a set of high-dimensional objects $V$, $|V|=n$, the high-dimensional probability $p_{ij}$ measuring the similarity between two objects $v_i$ and $ v_j$ is modeled as: 

\begin{equation}
    p_{ij} = \frac{p_{j|i} + p_{i|j}}{2n}
    \label{eq:pij}
\end{equation}

\noindent
computed using the conditional probability $p_{j|i}$: 

\begin{equation}
p_{j|i} = \exp \left ( -\frac{d^2_{ij}}{2\sigma^2_i} \right ) / \sum_{k \neq i} \exp \left ( -\frac{d^2_{ik}}{2\sigma^2_i} \right )
\label{eq:pijcond}
\end{equation}

\noindent
where $d_{ij}$ is the high-dimensional distance between $v_i$ and $v_j$, while the standard deviation $\sigma_i$ is selected in order to produce a \emph{perplexity} value of $u = 2^{-\sum_j p_{j|i} \log_2 p_{j|i}}$, for a constant $u$. 

The low-dimensional probability $q_{ij}$ modeling the distance between the projections of $v_i$ and $v_j$ in low-dimensional space, $X_i$ and $X_j$ respectively, is computed as:

\[
q_{ij} = q_{ji} = \frac{(1 + || X_i - X_j||^2)^{-1}}{ \sum_{k \neq l} (1 + || X_k - X_l||^2)^{-1}}
\]

The projection is computed by minimizing the \emph{KL (Kullback-Leibler)}~\cite{kullback1997information} divergence between the distributions of $p_{ij}$ and $q_{ij}$:

\begin{equation}
    KL = \sum_{i \neq j} p_{ij} \log \frac{p_{ij}}{q_{ij}}
    \label{eq:kl}
\end{equation}

\noindent
where KL divergence is minimized using gradient descent.

The second term $C_{CMP}$ is the \emph{compression}~\cite{van2008visualizing}, used to speed up convergence, computed as the average of the squared distance between each $X_i$, the position of a vertex $v_i$ in the drawing, and the origin.

The third term $C_{ENT}$ is \emph{entropy}~\cite{gansner2012maxent}, added to avoid clutter. 
It sums up the logarithm of the distance between each pair of vertices in the drawing with a small regularization constant $\epsilon$ added. When $\epsilon = 0$, the partial derivative with respect to a vertex $v_k$ is computed as: 

\begin{equation}
\frac{\delta}{\delta v_k}(C_{ENT}) = \frac{\lambda_r}{n^2} \sum_{i \neq k} \frac{X_k - X_i}{||X_k - X_i||^2}
\label{eq:entgrad}
\end{equation}
Multipliers $\lambda_{KL}, \lambda_c, \lambda_r$ are added to control the contributions of each term in different stages of the gradient descent process.

Overall, \texttt{tsNET} runs in $O(n^3)$ time, as the computation of high-dimensional probabilities, KL divergence, compression, and entropy terms takes $O(n^3)$, $O(n^2)$ and $O(n^2)$ time, respectively. 
Specifically, \texttt{t-SNE} runs in  $O(n^2)$-time for computing the distances between each pair of points in both the high- and low-dimensional spaces. 

\texttt{tsNET} has been evaluated against well-known graph layouts, such as SFDP~\cite{hu2005efficient}, LinLog~\cite{noack2003energy}, GRIP~\cite{gajer2002grip}, and NEATO~\cite{north2004drawing}, where it outperforms the other layouts on neighborhood preservation~\cite{kruiger2017graph}.
%
In this paper, we use \texttt{tsNET} to refer to $tsNET*$~\cite{kruiger2017graph}, denoting \texttt{tsNET} with a Pivot MDS (PMDS)~\cite{brandes2006eigensolver} initialization, for simplicity.

\subsection{Barnes-Hut SNE}


\texttt{BH-SNE} \emph{(Barnes-Hut SNE)}~\cite{van2014accelerating} reduces the runtime to $O(n \log n)$~\cite{van2014accelerating} using $k$NN and quadtree-based approximation. It first reformulates the KL divergence gradient as:

\begin{equation}
4 \left ( \sum_{j \neq i} p_{ij}q_{ij} Z(X_i - X_j) - \sum_{j \neq i} q^2_{ij} Z(X_i - X_j) \right )
\label{eq:klgrad}
\end{equation}

where $Z = \sum_{k\neq l} (1 + || X_k - X_l||^2)^{-1}$ is the normalization term. 
Compared to the {\em force-directed algorithm}~\cite{eades1984heuristic,battista1998graph}, the most popular graph drawing algorithm which consists of computing \emph{repulsion} and \emph{attraction} forces, the first term of Eq. \ref{eq:klgrad} is analogous to attraction forces, while the second term is analogous to repulsion forces.

The first ``attraction'' term of Eq. \ref{eq:klgrad} is reduced to $O(n \log n)$ time by computing the $k$NN of each object in high-dimensional space using vantage-point trees~\cite{yianilos1993data}. 
Then $p_{j|i}$ (Eq. \ref{eq:pijcond}) is computed only between an object $v_i$ and each object in $N_k(v_i)$, the set of its $k$ closest neighbors, and set to zero between $v_i$ and objects not part of $N_k(v_i)$.

The second ``repulsion'' term of Eq. \ref{eq:klgrad} is reduced by using \emph{quadtrees}, based on the Barnes-Hut algorithm~\cite{barnes1986hierarchical}. 
Specifically, a decomposition of the low-dimensional space is computed using a quadtree. For each point $X_i$, a depth-first search of the quadtree is performed starting from the root, and each quadtree node is checked to see if the cell can be used as a ``summary'' for all its constituent points; if so, the children of the node are not traversed and the low-dimensional probabilities between $X_i$ and all constituent points of the quadtree node are approximated using the center of mass of the quadtree cell.

\subsection{FFT-accelerated Interpolation-based t-SNE}

A more recent approach called \texttt{FIt-SNE} \emph{(FFT-accelerated Interpolation-based t-SNE)}~\cite{linderman2019fast} reduces the runtime of computing the second term of Eq. \ref{eq:klgrad} to $O(n)$ time using an interpolation-based method. 
Specifically, for a projection to two dimensions, the second term of Eq. \ref{eq:klgrad} can be expressed as four sums multiplying the low-dimensional probabilities with a kernel $K$, in the form of:

\[
    \sum^n_{j=1} K(X_i,X_j)q_{ij},
    \]
    \[ K(X_i,X_j) = \frac{1}{(1+||X_i - X_j||^2)} \text{ or } \frac{1}{(1+||X_i - X_j||^2)^{2}}
\]

By dividing the low-dimensional projection area into a constant number of $I$ ``intervals'' and defining a constant number of $P$ interpolant points inside each interval, the kernel $K$ computed between two intervals can be approximated in $O(n)$ time using the interpolant points instead of using the actual projected points. 
The computation of the interpolation is accelerated using FFT (Fast Fourier Transform)~\cite{nussbaumer1981fast}.

In addition, \texttt{FIt-SNE} accelerates the computation of high-dimensional probabilities by using \emph{approximate} 
$k$-nearest neighbors~\cite{bernhardsson2013annoy,muja2009fast}. While the practical runtime is often faster than exact $k$NN in practice, the worst-case runtime complexity is still $O(n \log n)$, resulting total runtime of \texttt{FIt-SNE} to still be $O(n \log n)$.

\section{\texttt{BH-tsNET}, \texttt{FIt-tsNET}, and \texttt{L-tsNET} Algorithms}
\label{sec:algo}

We now present  \texttt{BH-tsNET}, \texttt{FIt-tsNET}, and \texttt{L-tsNET}, improving the $O(n^3)$ runtime of \texttt{tsNET}
to $O(n \log n)$ and $O(n)$ time, respectively. 
To reduce the runtime of \texttt{tsNET}, we need to reduce the runtime of the following three components, each running in $O(n^3)$ or $O(n^2)$ time:

\begin{itemize}
    \item \textbf{C0:} Computation of high-dimensional probabilities (Eq. \ref{eq:pij}-\ref{eq:pijcond})
    \label{item:pij}
    \item \textbf{C1:} Computation of KL divergence gradient ($C_{KL}$ from Eq. \ref{eq:tsnet}) \label{item:kl}
    \item \textbf{C2:} Computation of entropy ($C_{ENT}$ from Eq. \ref{eq:tsnet}) \label{item:ent}
\end{itemize}

\emph{C0}, the computation of $p_{ij}$ (Eq. \ref{eq:pij}), can be considered as a pre-processing step, used in both \texttt{t-SNE} (Eq. \ref{eq:kl}) and \texttt{tsNET} (Eq. \ref{eq:tsnet}), as the high-dimensional probabilities $p_{ij}$ used in the first 
term of Eq. \ref{eq:tsnet} only needs to be computed once. 

\emph{C1} is part of the gradient descent loop (first 
term $C_{KL}$ of Eq. \ref{eq:tsnet}), used in both \texttt{t-SNE} and \texttt{tsNET}. 
Meanwhile, \emph{C2} is also part of the gradient descent loop (third term $C_{ENT}$ of Eq. \ref{eq:tsnet}), but only in \texttt{tsNET}. 
Note that compression computation (second term $C_{CMP}$ of Eq. \ref{eq:tsnet}) already runs in linear time, therefore no improvement is required.

For C1, the fast KL divergence gradient computation of fast \texttt{t-SNE} algorithms can be used. 
However, for C0, as \texttt{tsNET} uses shortest path distances to compute $p_{ij}$, reducing the runtime of C0 needs a new method. 
Likewise, a new method to reduce the runtime of C2 is needed as entropy is not included in fast \texttt{t-SNE} algorithms.

Therefore, we present our new fast \texttt{tsNET} algorithms, with our new approaches to reduce the runtimes of C0 and C2, integrated with  
fast \texttt{t-SNE} algorithms to reduce the runtime of C1.


\subsection{BH-tsNET}

\texttt{BH-tsNET} first uses our new partial BFS approach to compute the $k$NN of vertices in the graph for C0, reducing the runtime from $O(n^3)$ to $O(n)$. 
We then integrate the $O(n \log n)$ time quadtree approximation used by \texttt{BH-SNE} for C1. 
Finally, we use our new quadtree-based entropy for C2 to reduce the runtime from $O(n^2)$ to $O(n \log n)$.
See the details in Algorithm 
\ref{alg:bhtsnet}.

\begin{algorithm}
 \caption{\textbf{BH-tsNET}}
 \begin{algorithmic}[1]
 \State \textbf{Input:} Graph $G = (V,E)$, perplexity $u$.
 \State Compute initial layout $D$ of $G$ using Pivot MDS
 \State $k = 3u$ \label{line:knnprobstart}
 \For{$v \in V$}
    \State // C0: $p_{ij}$ computation
    \State Run BFS starting at $v$ until $k$ vertices are visited.
    \State $N_k(v)$: set of $k$NN of $v$ computed using BFS.
    \For{$u \in N_k(v)$}
        \State Compute $p_{u|v}$ using Eq. \ref{eq:pijcond}.
    \EndFor
 \EndFor \label{line:knnprobend}
 \State Compute $p_{ij}$ using Eq. \ref{eq:pij}.
 \For{iterations in 1 to 500}
    \State // C1: KL divergence computation
    \State Construct quadtree of drawing $D$.
    \State Compute $C_{KL}$ using quadtree approximation.
    \State Compute $C_{CMP}$.
    \State // C2: entropy computation
    \State Compute quadtree of vertices in $D$. \label{line:knnentstart}
    \State Compute $C_{ENT}$ using quadtree approximation. \label{line:knnentend}
    \State Move each vertex $v$ in the direction of $\frac{\delta}{\delta v}(C)$. \label{line:movegradient}
 \EndFor
 \State \textbf{return} $D$
 \end{algorithmic}
 \label{alg:bhtsnet}
\end{algorithm}

For C0, \texttt{tsNET} requires the computation of shortest paths between every pair of vertices. \texttt{BH-SNE} reduces the runtime of C0 by computing the $k$NN of each object using vantage-point trees; however, this approach cannot be directly applied as-is for graphs.

Therefore, we design \texttt{BH-tsNET} by adding a new \emph{partial BFS (Breadth-First Search)} starting from each vertex $v$ in a graph $G = (V, E)$, stopping once $k$ other vertices have been visited; when multiple vertices are of the same shortest path distance from $v$, ties are broken randomly. 
Then $p_{ij}$ can be computed in the same way as \texttt{BH-SNE}, using the $k$ vertices selected in the partial BFS as the $k$-nearest neighbors of the vertex $v$. 
We set $k = 3u$ for perplexity $u$, same as \texttt{BH-SNE}~\cite{van2014accelerating}. 
See lines \ref{line:knnprobstart}-\ref{line:knnprobend} in Algorithm \ref{alg:bhtsnet} for details.

The runtime of C1 is reduced to $O(n \log n)$ time by using the quadtree approach of \texttt{BH-SNE}: in each gradient descent iteration, a quadtree is computed on the drawing $D$, and the repulsion term of the KL gradient is approximated using the quadtree.

For C2, we newly reduce the runtime of entropy computation, by similarly utilizing quadtrees: for each vertex $v_i$ with position $X_i$ in drawing $D$, perform a depth-first search of the quadtree computed on $D$. If a quadtree node $N_j$ can be used as a summary, compute the entropy gradient between $v_i$ and $N_j$ as:

\begin{equation} 
    \frac{\lambda_r}{n^2} \left ( \frac{|N_j|(X_Nj - X_i)}{\epsilon - ||X_Nj - X_i||^2} \right )
    \label{eq:qtent}
\end{equation}

where $X_Nj$ is the coordinates of the center of mass of $N_j$ and $|N_j|$ is the number of constituent points of $N_j$. With all components computed, each vertex $v$ is moved in the direction of the partial derivative with regards to $v$ of the cost function $C$, see line \ref{line:movegradient}.

The following theorem proves the time and space complexity of \texttt{BH-tsNET}.

\begin{theorem} \label{theorem:bhtsnet}
    \texttt{BH-tsNET} runs in $O(n \log n)$ time, with $O(n \log n + m)$ space complexity.
\end{theorem}

\begin{proof}
The runtime of BH-tsNET can be analyzed using the runtime of each step C0-C2. In C0, partial BFS runs once per vertex, where the traversal stops after $k$ vertices are visited; as $k$ is constant with respect to $n$, C0 runs in $O(n)$ time in total. 
Both C1 and C2 run in $O(n \log n)$ time, as the quadtree can be computed in $O(n \log n)$ time. Therefore, overall, \texttt{BH-tsNET} runs in $O(n \log n)$ time. 

On top of the input graph requiring $O(n + m)$ space (for $n$ vertices and $m$ edges), $O(n)$ space is needed to store the $p_{ij}$ values (constant $k$ values per vertex), and $O(n \log n)$ space to store the quadtree structure, for a total of $O(n \log n + m)$ space.
\end{proof}

\subsection{FIt-tsNET}

In \texttt{FIt-tsNET}, we first use our new partial BFS approach to compute the $k$NN of vertices in the graph for C0, reducing the runtime from $O(n^3)$ to $O(n)$. 
We then use the FFT-accelerated interpolation of \texttt{FIt-SNE} for C1 to reduce the runtime from $O(n^2)$ to $O(n)$. 
Finally, we use our new quadtree-based entropy for C2 to reduce the runtime from $O(n^2)$ to $O(n \log n)$.
See the details in Algorithm 
\ref{alg:fittsnet}. 

\begin{algorithm}
 \caption{\textbf{FIt-tsNET}}
 \begin{algorithmic}[1]
 \State \textbf{Input:} Graph $G = (V,E)$, perplexity $u$, \# of intervals $I$, \# of interpolation points $P$
 \State Compute initial layout $D$ of $G$ using Pivot MDS.
 \State  $k = 3u$
 \For{$v \in V$}
 \State  // C0: $p_{ij}$ computation
    \State Run BFS starting at $v$ until $k$ vertices are visited.
    \State $N_k(v)$: set of $k$NN of $v$ computed using BFS.
    \For{$u \in N_k(v)$}
        \State Compute $p_{u|v}$ using Eq. \ref{eq:pijcond}.
    \EndFor
 \EndFor
 \State Compute $p_{ij}$ using Eq. \ref{eq:pij}.
 \For{iterations in 1 to 500}
    \State // C1: KL divergence computation
    \For{interval $i$ in 1 to $I$} \label{line:interpstart}
        \State Compute $P$ interpolation points for interval $i$.
    \EndFor
    \State Compute $C_{KL}$ sums using FFT-accelerated interpolation.\label{line:interpend}
    \State Compute $C_{CMP}$.
    \State // C2: entropy computation
    \State Compute quadtree of vertices in $D$.
    \State Compute $C_{ENT}$ using quadtree approximation.
    \State Move each vertex $v$ in the direction of $\frac{\delta}{\delta v}(C)$.
 \EndFor
 \State \textbf{return} $D$
 \end{algorithmic}
 \label{alg:fittsnet}
\end{algorithm}

For C0, \texttt{FIt-SNE} reduces the computation to $O(n \log n)$ time using approximate $k$NN; however, this also cannot be directly applied as-is for graphs.
Therefore, \texttt{FIt-tsNET} uses our new partial all-source BFS to compute the $k$NN of each vertex, similar to \texttt{BH-tsNET}.

Unlike \texttt{BH-tsNET}, \texttt{FIt-tsNET} further improves the runtime of C1 to $O(n)$ using the FFT-accelerated interpolation of \texttt{FIt-SNE} (see lines \ref{line:interpstart}-\ref{line:interpend}).
The number of intervals $I$ and the number of interpolation points $P$ are constants, independent of the number of vertices $n$~\cite{linderman2019fast}, keeping the runtime of C1 as $O(n)$. 

For C2, we reduce the runtime to $O(n \log n)$ using quadtrees (see Eq. \ref{eq:qtent}), similar to \texttt{BH-tsNET}.
The following theorem proves the time and space complexity of \texttt{FIt-tsNET}.

\begin{theorem} \label{theorem:fittsnet}
    \texttt{FIt-tsNET} runs in $O(n \log n)$ time with $O(n \log n + m)$ space complexity.
\end{theorem}

\begin{proof}
Both C0 and C1 take $O(n)$ time due to the partial BFS for C0 and the FFT-accelerated interpolation for C1. 
C2 runs in $O(n \log n)$ time due to the usage of quadtrees. Therefore, the overall runtime of \texttt{FIt-tsNET} is $O(n \log n)$. 

Similar to \texttt{BH-tsNET}, \texttt{FIt-tsNET} requires $O(n)$ space to store the $p_{ij}$ values and $O(n \log n)$ space for the quadtree, while the interpolation points only require $O(1)$ space, 
therefore in total $O(n \log n + m)$ space.
\end{proof}


Note that the overall time complexity of \texttt{FIt-tsNET} and \texttt{BH-tsNET} is the same, $O(n \log n)$.
However, we expect \texttt{FIt-tsNET} to run faster than \texttt{BH-tsNET} in practice, as C1 takes $O(n)$ time by \texttt{FIt-tsNET}, faster than $O(n \log n)$ time by \texttt{BH-tsNET}.

\subsection{L-tsNET (Linear-tsNET)}

In \texttt{L-tsNET}, we first use our new partial BFS approach to compute the $k$NN of vertices in the graph for C0, reducing the runtime from $O(n^3)$ to $O(n)$. 
We then use the FFT-accelerated interpolation of \texttt{FIt-SNE} for C1 to reduce the runtime from $O(n^2)$ to $O(n)$. 
Finally, we use our new interpolation-based entropy for C2 to reduce the runtime from $O(n^2)$ to $O(n)$.
See the details in Algorithm  
\ref{alg:ltsnet}.

\begin{algorithm}
 \caption{\textbf{L-tsNET}}
 \begin{algorithmic}[1]
 \State \textbf{Input:} Graph $G = (V,E)$, perplexity $u$, \# of intervals $I$, \# of interpolation points $P$
 \State Compute initial layout $D$ of $G$ using Pivot MDS.
 \State $k = 3u$ 
 \For{$v \in V$}
 \State // C0: $p_{ij}$ computation
    \State Run BFS starting at $v$ until $k$ vertices are visited.
    \State $N_k(v)$: set of $k$NN of $v$ computed using BFS.
    \For{$u \in N_k(v)$}
        \State Compute $p_{u|v}$ using Eq. \ref{eq:pijcond}.
    \EndFor
 \EndFor 
 \State Compute $p_{ij}$ using Eq. \ref{eq:pij}.
 \For{iterations in 1 to 500}
    \State // C1: KL divergence computation
    \For{interval $i$ in 1 to $I$}
        \State Compute $P$ interpolation points for interval $i$.
    \EndFor
    \State Compute the sums for $C_{KL}$ using FFT-accelerated interpolation.
    \State Compute $C_{CMP}$.
    \State \// C2: entropy computation
    \State Compute sums $h_1,h_2,h_3$ using FFT-accelerated interpolation.
    \State Compute gradient of $C_{ENT}$ with Eq. \ref{eq:ltsnet_entgrad}.
    \State Move each vertex $v$ in the direction of $\frac{\delta}{\delta v}(C)$.
 \EndFor
 \State \textbf{return} $D$
 \end{algorithmic}
  \label{alg:ltsnet}
\end{algorithm}

For C0 and C1, the runtime is reduced to $O(n)$ using partial all-source BFS and FFT-accelerated interpolation, respectively, the same as \texttt{FIt-SNE}. 
Unlike \texttt{FIt-tsNET}, C2 is further reduced to $O(n)$ time by extending the FFT-accelerated interpolation to entropy computation.

Specifically, to apply the interpolation method used in \texttt{FIt-SNE} to entropy computation, we reformulate the entropy partial derivative (Eq. \ref{eq:entgrad}) to incorporate a nonzero regularization constant $\epsilon$:

\[
 \frac{\lambda_r}{n^2} \sum_{i \neq k} \frac{X_k - X_i}{\epsilon - ||X_k - X_i||^2}
\]

Note that by reformulating the summation as $-\frac{\lambda_r}{n^2} \sum_{i \neq k} \frac{X_k - X_i}{\epsilon + ||X_k + X_i||^2}$, the denominator is now in a similar form to one of the kernels used in \texttt{FIt-SNE}, $\frac{1}{1 + ||X_k - X_i||^2}$. 
We define the entropy kernel as $K_e = \frac{1}{\epsilon + ||X_k - X_i||^2}$, and define $h_1, h_2, h_3$ as sums of the form:

\[
h_1 = \sum_{i \neq k} \frac{1}{\epsilon + ||X_k - X_i||^2}, \\
h_2 = \sum_{i \neq k} \frac{x_k}{\epsilon + ||X_k - X_i||^2},
\]

\[
h_3 = \sum_{i \neq k} \frac{y_k}{\epsilon + ||X_k - X_i||^2}
\]

\noindent where $x_k$ and $y_k$ are the $x-$ and $y-$coordinate respectively of $v_k$ in the drawing.

The $x$-component of the partial derivative of the entropy with respect to a vertex $v_i$ can then be restated as:

\begin{equation}
 -\frac{\lambda_r}{n^2} \sum_{i \neq k} \frac{x_k - x_i}{\epsilon + ||X_k - X_i||^2} = -\frac{\lambda_r}{n^2} \left ( h_2 - x_i h_1 \right )
 \label{eq:ltsnet_entgrad}
\end{equation}

The $y$-component can then be computed in a similar way, substituting $h_3$ for $h_2$. 
In our algorithms, we set $\epsilon = \frac{1}{20}$, the same as \texttt{tsNET}~\cite{kruiger2017graph}.
In this form, the FFT-accelerated interpolation can be used to compute the sums $h_1, h_2, h_3$ to perform the entropy computation in C2 in $O(n)$ time, similar to C1.

The following theorem proves the time and space complexity of \texttt{L-tsNET}, the fastest and the most space-efficient algorithm.

\begin{theorem} \label{theorem:ltsnet}
    \texttt{L-tsNET} runs in $O(n)$ time with $O(n + m)$ space complexity.
\end{theorem}

\begin{proof}
Both steps for C0 and C1 run in $O(n)$ time using the partial BFS and FFT-accelerated interpolation, respectively. 
C2 now runs in $O(n)$ time by using the FFT-acceleration interpolation for the entropy computation. Overall, \texttt{L-tsNET} takes $O(n)$ time. 

For space complexity, both C1 and C2 now only require the storage of a constant number of interpolation points, thus reducing the space complexity to $O(n + m)$ in total, compared to the $O(n \log n + m)$ space complexity of \texttt{BH-tsNET} and \texttt{FIt-tsNET}.
\end{proof}

\section{Experiment 1: \texttt{BH-tsNET}, \texttt{FIt-tsNET}, \texttt{L-tsNET} vs. \texttt{tsNET}}
\label{sec:eval}

We now present experiments comparing our algorithms, \texttt{BH-tsNET}, \texttt{FIt-tsNET} and \texttt{L-tsNET}, to \texttt{tsNET}. 
Since the main scientific contribution of our new algorithms focuses on improving the runtime efficiency of \texttt{tsNET}, we mainly focus on the runtime and quality comparison with \texttt{tsNET}.

Note that \texttt{tsNET} was shown to outperform well-known graph layouts (\texttt{SFDP}~\cite{hu2005efficient}, \texttt{LinLog}~\cite{noack2003energy}, \texttt{GRIP}~\cite{gajer2002grip},  \texttt{NEATO}~\cite{north2004drawing}) on neighborhood preservation~\cite{kruiger2017graph}. Therefore, we do not consider these layouts for comparison.
Moreover, as our main aim is to reduce the \emph{theoretical time complexity} of tsNET, we do not consider parallel \texttt{t-SNE} methods using GPU (\texttt{t-SNE-CUDA}~\cite{chan2018tsne},  \texttt{GPGPU-SNE}~\cite{pezzotti2020gpgpu}).

\subsection{Implementation and Experiment Design}

\noindent{\bf Implementation.}
We implement \texttt{BH-tsNET}, \texttt{FIt-tsNET}, and \texttt{L-tsNET} in Python and C++.
Specifically, the computation of high-dimensional probabilities is implemented in Python based on the \texttt{t-SNE} implementation from the scikit-learn library~\cite{pedregosa2011scikit}, while all other functions are implemented in C++ based on the implementations of \texttt{BH-SNE}~\cite{van2014accelerating} and \texttt{FIt-SNE}~\cite{linderman2019fast}. 
In order to have a fair runtime comparison, we also implement \texttt{tsNET} in Python and C++ based on the same libraries.

\noindent{\bf Hypotheses.}
We compare our \texttt{BH-tsNET}, \texttt{FIt-tsNET}, and \texttt{L-tsNET} algorithms to \texttt{tsNET} on runtime, quality metrics (neighborhood preservation), and visual comparison. 

Based on the time complexity analyses in Theorems \ref{theorem:bhtsnet}, \ref{theorem:fittsnet}, and  \ref{theorem:ltsnet}, we expect that (1) all our algorithms run significantly faster than \texttt{tsNET}, (2) \texttt{L-tsNET} runs the fastest, and (3) \texttt{FIt-tsNET} runs faster than \texttt{BH-tsNET} due to the runtime reduction of C1 to $O(n)$.


\begin{hyp} \label{hyp:runtime_base}
\texttt{BH-tsNET}, \texttt{FIt-tsNET}, and \texttt{L-tsNET} run significantly faster than \texttt{tsNET}.
\end{hyp}

\begin{hyp} \label{hyp:runtime_comp}
\texttt{FIt-tsNET} runs faster than \texttt{BH-tsNET}.
\end{hyp}

\begin{hyp} \label{hyp:runtime_lts}
\texttt{L-tsNET} runs significantly faster than  \texttt{FIt-tsNET}.
\end{hyp}

\begin{hyp} \label{hyp:quality}
\texttt{BH-tsNET}, \texttt{FIt-tsNET}, and \texttt{L-tsNET} compute similar quality drawings to \texttt{tsNET}.
\end{hyp}

\noindent{\bf Benchmark Data Sets.}
We use well-known benchmark data sets used in \texttt{tsNET}~\cite{kruiger2017graph} and standard large and complex graph test suites used by the latest fast graph drawing algorithms~\cite{meidiana2020sublinear,meidiana2024sublinearforce,meidiana2021stress}: (1) real-world scale-free graphs~\cite{snapnets}, globally sparse but locally dense graphs with high clustering and small diameter; (2) $GION$ data set of large RNA networks~\cite{marner2014gion}, globally sparse but locally dense graphs with long diameters; and (3) mesh graphs~\cite{davis2011university}. 
Table \ref{table:data} in Supplementary Materials shows the details of data sets, including the density and perplexity used.

\noindent{\bf Experiment Details.}
To keep our implementation and evaluation consistent with the \texttt{tsNET} experiment~\cite{kruiger2017graph}, we set the perplexity $u$ as 40 by default, as used in~\cite{kruiger2017graph}.
When a perplexity of 40 results in a ``perplexity too low'' error on the implementation from~\cite{kruiger2017graph}, we set the perplexity as the lowest multiple of 100 that does not produce the error.

For each graph, we run \texttt{tsNET},  \texttt{tsNET}, \texttt{BH-tsNET}, and \texttt{FIt-tsNET} five times and take the average for runtime, to ensure consistency.
Moreover, we use Pivot MDS initialization, as was shown to obtain better quality than random initialization for \texttt{tsNET}~\cite{kruiger2017graph}. 
For each algorithm running on one graph, the quality metrics are the same between different runs as \texttt{tsNET} is deterministic given the same initial layout.

\begin{figure}[h]
    \centering
    \subfloat[Runtime. avg.]{
    \includegraphics[width=0.4\columnwidth]{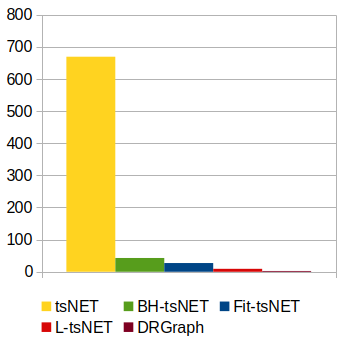}
    }
    \subfloat[Neigh. pres. avg.]{
    \includegraphics[width=0.4\columnwidth]{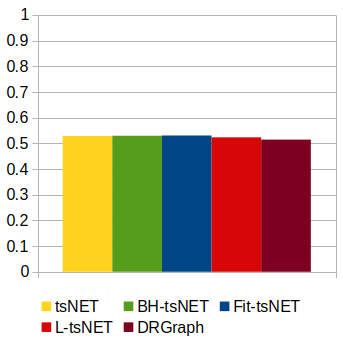}
    }
    \qquad
    \subfloat[Stress avg.]{
    \includegraphics[width=0.4\columnwidth]{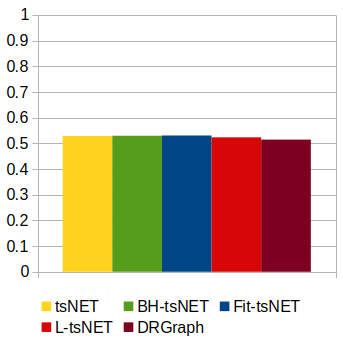}
    }
    \subfloat[Shape-based avg.]{
    \includegraphics[width=0.4\columnwidth]{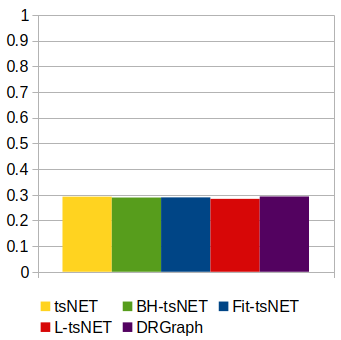}
    }
    \qquad
    \subfloat[Edge cross. avg. (all)]{\quad 
    \includegraphics[width=0.4\columnwidth]{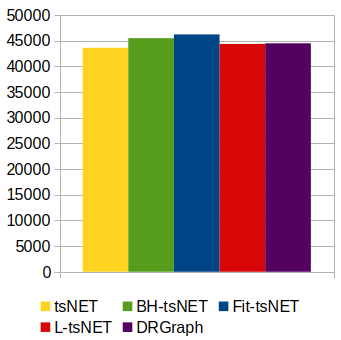} \quad}
    \subfloat[Edge cross. avg. (mesh)]{
    \includegraphics[width=0.4\columnwidth]{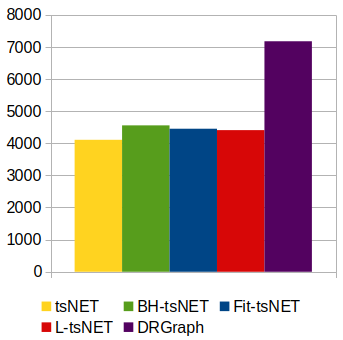} \quad}
    \caption{ Average runtime and quality metrics comparison of \texttt{tsNET}, \texttt{BH-tsNET}, \texttt{FIt-tsNET}, \texttt{L-tsNET}, and \texttt{DRGraph}. 
    \texttt{BH-tsNET}, \texttt{FIt-tsNET}, and \texttt{L-tsNET} all run significantly faster than \texttt{tsNET}, with very similar quality metrics. 
    On mesh graphs, \texttt{L-tsNET}  obtain better edge crossing than \texttt{DRGraph}.}
    \label{fig:metrics_drg_avg_all}
\end{figure}

\subsection{Runtime Comparison}

Figure \ref{fig:metrics_drg_avg_all}(a) shows the average runtime comparison between \texttt{BH-tsNET}, \texttt{FIt-tsNET}, and \texttt{L-tsNET}, confirming Hypothesis \ref {hyp:runtime_base}.
On average, \texttt{BH-tsNET}, \texttt{FIt-tsNET}, and \texttt{L-tsNET} run 93.5\%, 96\%, and 98.6\% faster \texttt{tsNET}, respectively, see Figure \ref{fig:tsnet_runtime}  for the runtime of the algorithms with all graphs in seconds. 
For larger graphs, the quadratic growth in runtime by \texttt{tsNET} dwarfs that of all of \texttt{BH-tsNET}, \texttt{FIt-tsNET}, and \texttt{L-tsNET}.

Figure \ref{fig:tsnet_runtime_nobase}  shows a detailed runtime comparison only between our three algorithms. 
On average, \texttt{FIt-tsNET} runs 36.7\% faster than \texttt{BH-tsNET}, confirming Hypothesis \ref{hyp:runtime_comp}. 
Moreover, \texttt{L-tsNET} runs 78.6\% and 66.3\% faster than \texttt{BH-tsNET} and \texttt{FIt-tsNET} respectively on average, confirming Hypothesis \ref{hyp:runtime_lts}.
Therefore, these results support Hypotheses \ref{hyp:runtime_base}, \ref{hyp:runtime_comp}, and \ref{hyp:runtime_lts}. 
Furthermore, Theorems \ref{theorem:bhtsnet}, \ref{theorem:fittsnet}, and  \ref{theorem:ltsnet} prove that such runtime comparison results seen in Figure \ref{fig:tsnet_runtime} also hold for larger graphs.

For smaller graphs with $|V|< 1700$, \texttt{BH-tsNET}, \texttt{FIt-tsNET}, and \texttt{L-tsNET} take almost the same runtime as \texttt{tsNET}, due to the overhead of computing quadtrees and setting up the interpolation.
However, for large graphs
\texttt{BH-tsNET}, \texttt{FIt-tsNET}, and \texttt{L-tsNET} run significantly faster than \texttt{tsNET} as the sizes of the graphs grow since such overheads become proportionally smaller compared to the significant speed-up achieved by runtime reduction of the gradient descent iteration.

\emph{In summary, our \texttt{BH-tsNET}, \texttt{FIt-tsNET}, and \texttt{L-tsNET} algorithms run significantly faster than \texttt{tsNET}, on average 93.5\%, 96\%, and 98.6\% faster respectively, supporting Hypotheses \ref{hyp:runtime_base} to \ref{hyp:runtime_lts}.}

\begin{figure*}[h!]
    \centering
    \includegraphics[width=0.93\textwidth]{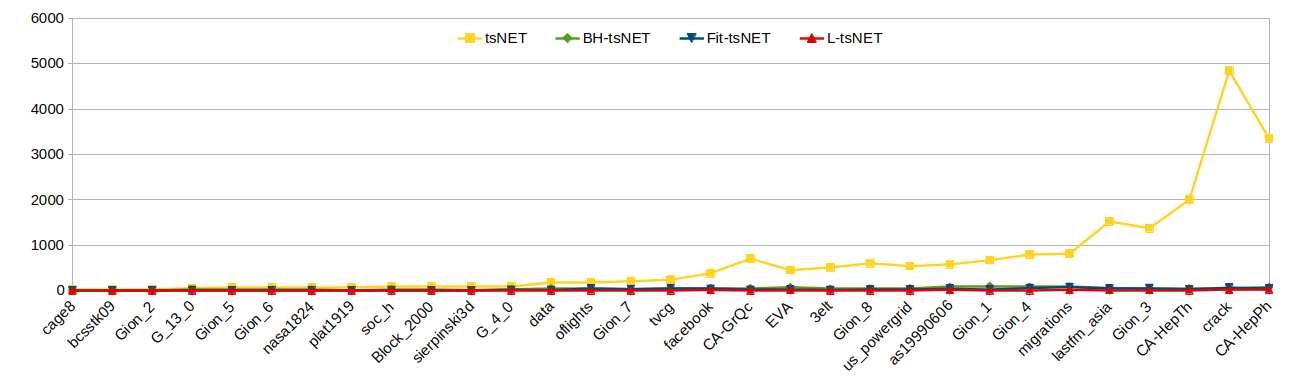}
    \caption{Runtime (in seconds): On average, all of \texttt{BH-tsNET}, \texttt{FIt-tsNET}, and \texttt{L-tsNET} run 93\%, 96\%, and 98.6\% faster than \texttt{tsNET} respectively.}
    \label{fig:tsnet_runtime}
  \vspace{-3mm}
\end{figure*}

\begin{figure*}[h!]
    \centering
    \includegraphics[width=0.93\textwidth]{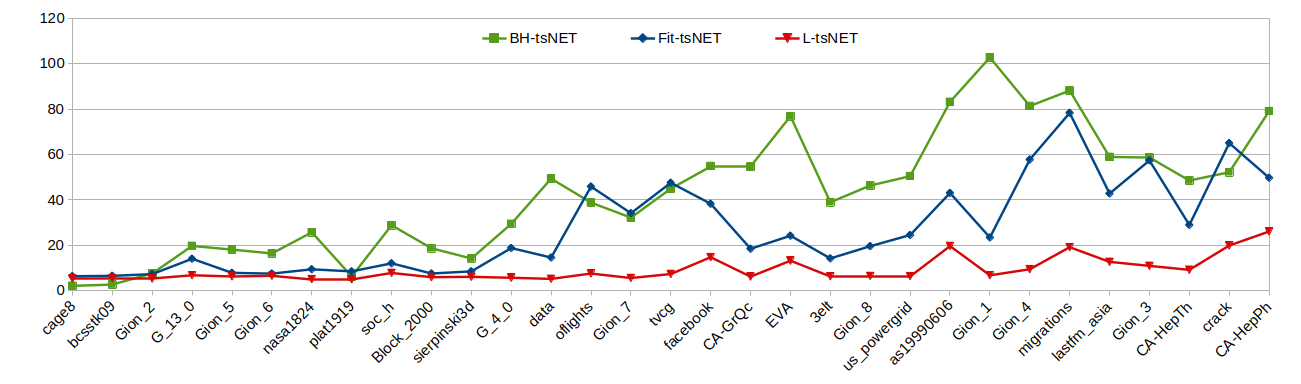}
    \caption{Runtime comparison (in seconds) with \texttt{tsNET} redacted. \texttt{L-tsNET} run on average 78.6\% and 66.3\% faster than \texttt{BH-tsNET} and \texttt{FIt-tsNET} respectively.}
    \label{fig:tsnet_runtime_nobase}
    \vspace{-3mm}
\end{figure*}

\subsection{Quality Metrics Comparison}

For quality metrics, we use the same metrics used in \texttt{tsNET} evaluation~\cite{kruiger2017graph},  {\em neighborhood preservation}~\cite{gansner2012maxent,kruiger2017graph} and {\em stress}~\cite{battista1998graph}.
In addition, we also most popular {\em readability} metrics for graph drawing, {\em edge crossing}~\cite{battista1998graph}, as well as the {\em faithfulness } metrics designed for large graphs, {\em shape-based metrics}~\cite{eades2017shape}.

\subsubsection{Neighborhood preservation (NP)} 

NP measures how similar the neighborhood $N_G(v, r)$, i.e., the set of vertices with a shortest path distance of at most $r$ from a vertex $v \in V$ in a graph $G = (V, E)$ (we set $r=2$, the same as~\cite{kruiger2017graph}) 
is represented by the \emph{geometric} neighborhood $N_D(v, r)$ of $v$ in a drawing $D$ of $G$ 
(i.e., the $|N_G(v, r)|$-closest neighbors of $v$ in $D$)~\cite{martins2015explaining}. 
Specifically, NP is computed as: 
\[
    \sum_{v \in V} \frac{|N_G(v,r) \cap N_D(v,r)|}{|N_G(v,r) \cup N_D(v,r)|}.
\]

Surprisingly, despite the much faster runtime, the NP metrics of the drawings computed by \texttt{BH-tsNET}, \texttt{FIt-tsNET}, and \texttt{L-tsNET} are almost the same as that of \texttt{tsNET} for all graphs, confirming Hypothesis \ref{hyp:quality}. 
See Figure \ref{fig:metrics_drg_avg_all} (b) for an average and Figure \ref{fig:tsnet_np} for details per data. NP has a value between 0 and 1, where a higher value is better.

Furthermore,  \texttt{L-tsNET} produce drawings with better NP metrics than \texttt{tsNET} for dense graphs with long diameters, such as GION\_1 and GION\_3,
as seen in Figure \ref{fig:tsnet_np} in Supplementary Materials.

\subsubsection{Stress (ST)} 

ST measures how proportionally the shortest path distance $d(v_i,v_j)$ between vertices $v_i$ and $v_j$ in $G$ is represented as the geometric distances $||X_i - X_j||$ between $X_i$ and $X_j$, the positions of the vertices in the drawing $D$ of $G$~\cite{battista1998graph,kamada1989algorithm,gansner2005graph}.
Specifically, we use the aggregated stress: 
\[
 \frac{1}{n(n-1)}\sum_{i, j \in |V|} \left ( \frac{d(v_i,v_j) - ||X_i - X_j||}{(d(v_i,v_j)} \right )^2   
\]

The ST metrics of the drawings computed by \texttt{BH-tsNET}, \texttt{FIt-tsNET}, and \texttt{L-tsNET} are also almost the same as the drawings computed by \texttt{tsNET}, confirming Hypothesis \ref{hyp:quality}.
See the average in Figure \ref{fig:metrics_drg_avg_all} (c) and Figure \ref{fig:tsnet_stress} in Supplementary Materials for the details per data. For ST, a lower value is better.

Surprisingly, for large dense graphs with large diameters such as GION graphs (GION\_8, GION\_1, GION\_4, and GION\_3), \texttt{L-tsNET} obtains drawings with significantly lower ST than \texttt{tsNET}, on average about 9\% lower.

\subsubsection{Shape-based Metrics (SB)}

SB~\cite{eades2017shape} are specifically designed to evaluate the quality of large graph drawings, measuring how faithfully the ``shape'' of the drawing, using the \emph{proximity graph}~\cite{toussaint1986computational} computed from a drawing $D$ of $G$, represents the ground truth structure of $G$. 
Specifically, we use the $dRNG$ (Degree-sensitive Relative Neighborhood Graph) proximity graph~\cite{hong2022dgg}, which was shown to measure SB more accurately than the  RNG used in ~\cite{eades2017shape}. 
SB has a value between 0 and 1, where a higher value is better.

Overall, the SB metrics of \texttt{BH-tsNET}, \texttt{FIt-tsNET}, and \texttt{L-tsNET} drawings are almost the same as \texttt{tsNET}, supporting Hypothesis \ref{hyp:quality}.
See Figure \ref{fig:metrics_drg_avg_all} (d) for the average SB metrics and Figure \ref{fig:tsnet_shp} in Supplementary Materials for the details per data.

More specifically, our algorithms tend to compute drawings with higher SB metrics than \texttt{tsNET} on dense real-world scale-free graphs (tvcg, CA-GrQc, CA-HepPh), and lower SB metrics on sparse real-world graphs (EVA and migrations).

\subsubsection{Edge Crossing (CR)}

CR is the most popular readability criterion to evaluate the quality of a graph drawing~\cite{battista1998graph,purchase1996validating}. 
The CR metrics of the drawings computed by \texttt{BH-tsNET}, \texttt{FIt-tsNET}, and \texttt{L-tsNET} are slightly higher than the drawings computed by \texttt{tsNET}, overall supporting Hypothesis \ref{hyp:quality}.
See Figure \ref{fig:metrics_drg_avg_all} (e) for the average CR and Figure \ref{fig:tsnet_crossing} in Supplementary Materials for the details per data.  A lower value for CR is better.

Surprisingly, \texttt{L-tsNET} computes drawings with almost the same CR metrics as drawings computed by \texttt{tsNET}, with 1.7\% difference on average, outperforming \texttt{BH-tsNET} and \texttt{FIt-tsNET} with  4.3\% and 6\% higher CR metrics respectively, still much lower than the significant runtime improvement of over 93\%.

\begin{table*}[h!]
    \centering
    \caption{Visual comparison: In general, \texttt{tsNET}, \texttt{BH-tsNET}, and \texttt{FIt-tsNET} produce drawings with visually the same quality.}
    \begin{tabular}{|c|c|c|c|}
    \hline
    tsNET & BH-tsNET & FIt-tsNET & L-tsNET \\ \hline
    \multicolumn{4}{|c|}{3elt} \\ \hline
    \includegraphics[width=0.33\columnwidth]{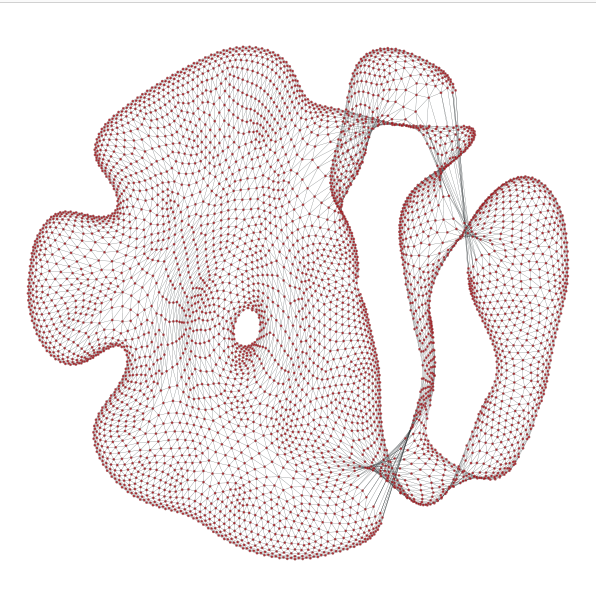} &
    \includegraphics[width=0.33\columnwidth]{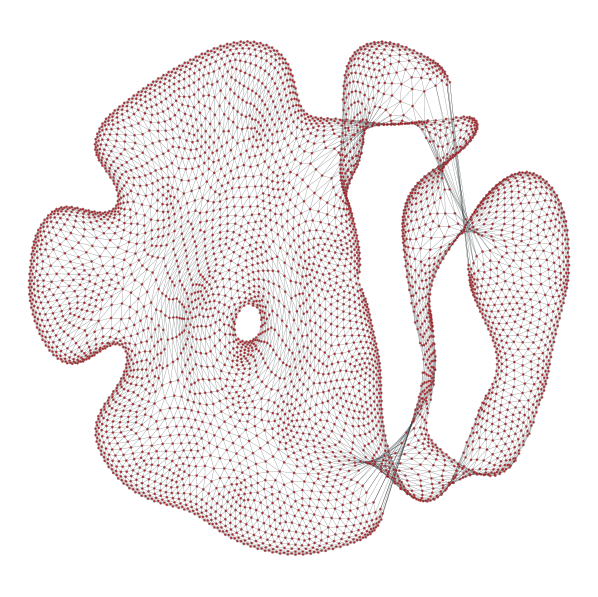} &
    \includegraphics[width=0.33\columnwidth]{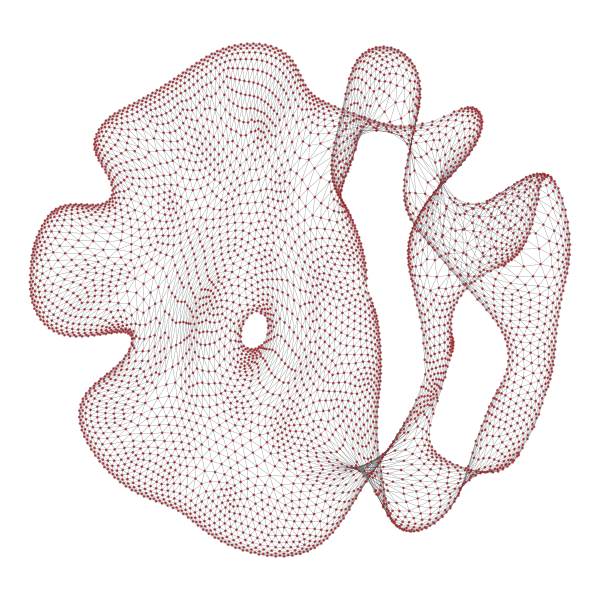} &
    \includegraphics[width=0.33\columnwidth]{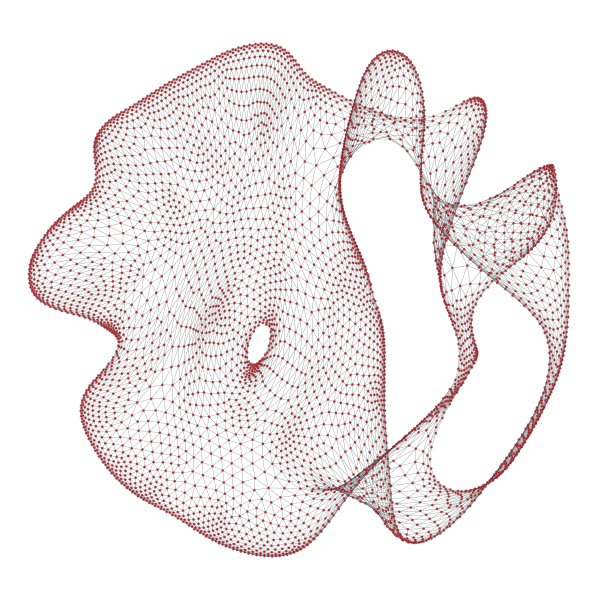} \\ \hline
    \multicolumn{4}{|c|}{oflights} \\ \hline
    \includegraphics[width=0.33\columnwidth]{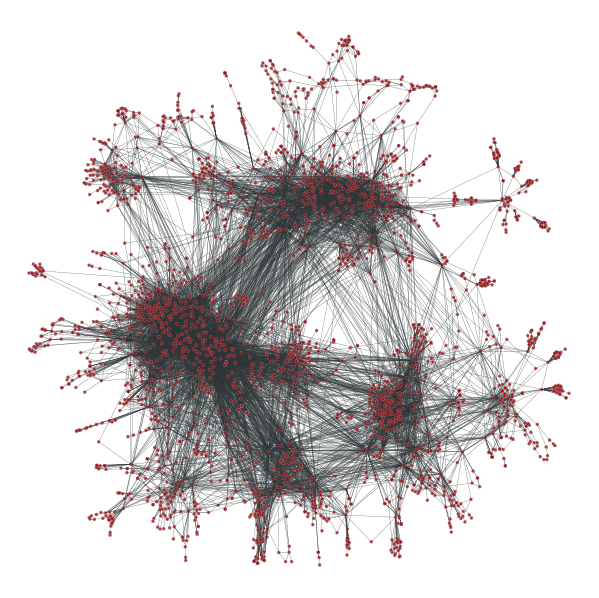} &
    \includegraphics[width=0.33\columnwidth]{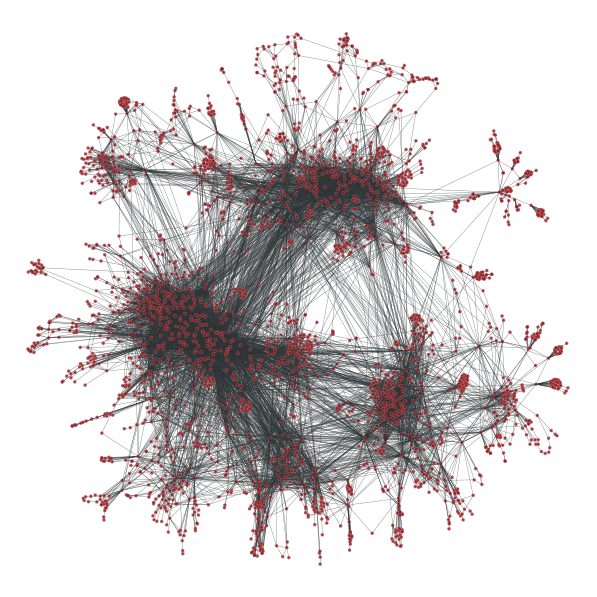} &
    \includegraphics[width=0.33\columnwidth]{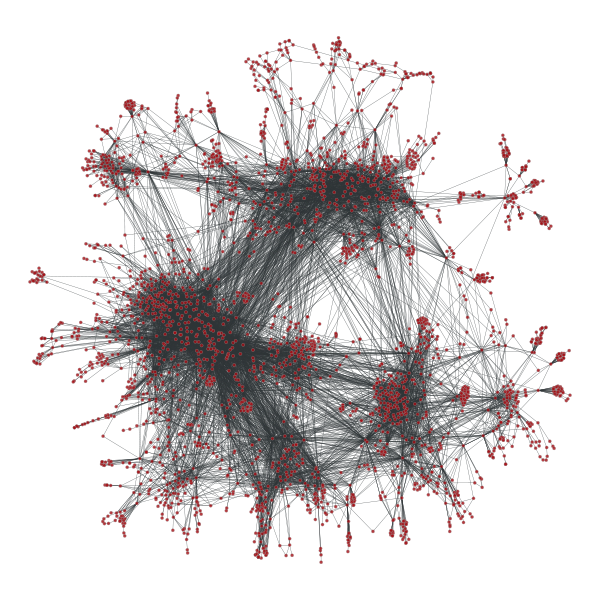} &
    \includegraphics[width=0.33\columnwidth]{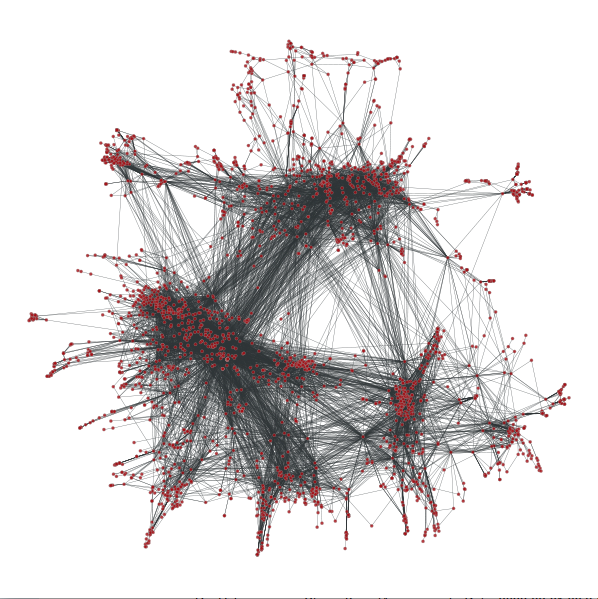} \\ \hline
    \multicolumn4{|c|}{GION\_7} \\ \hline
    \includegraphics[width=0.33\columnwidth]{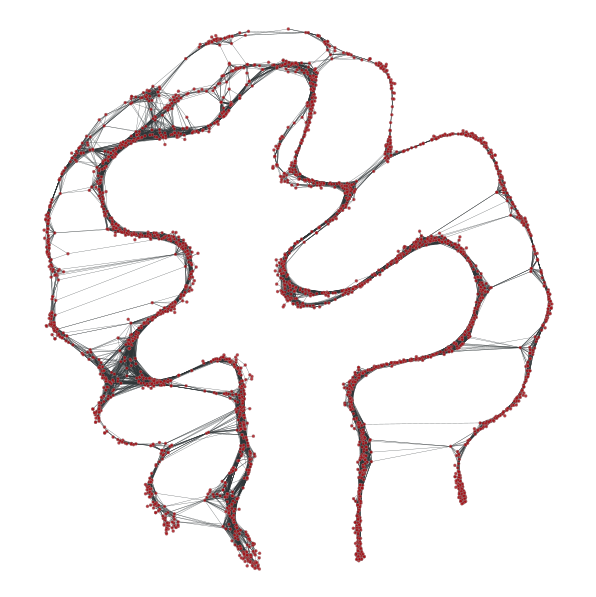} &
    \includegraphics[width=0.33\columnwidth]{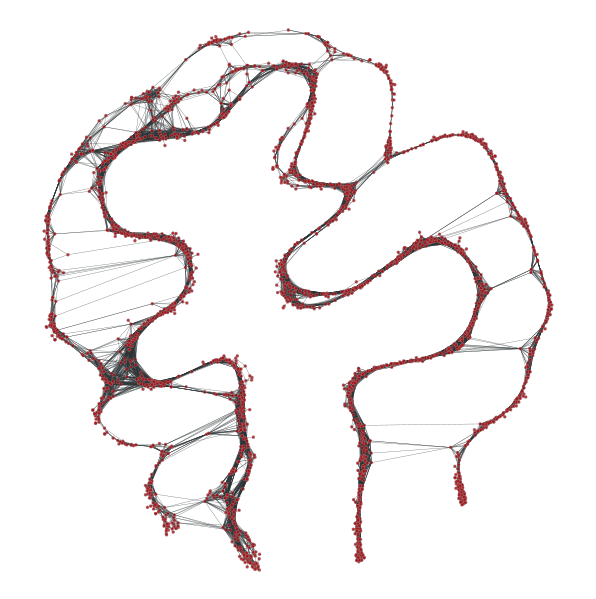} &
    \includegraphics[width=0.33\columnwidth]{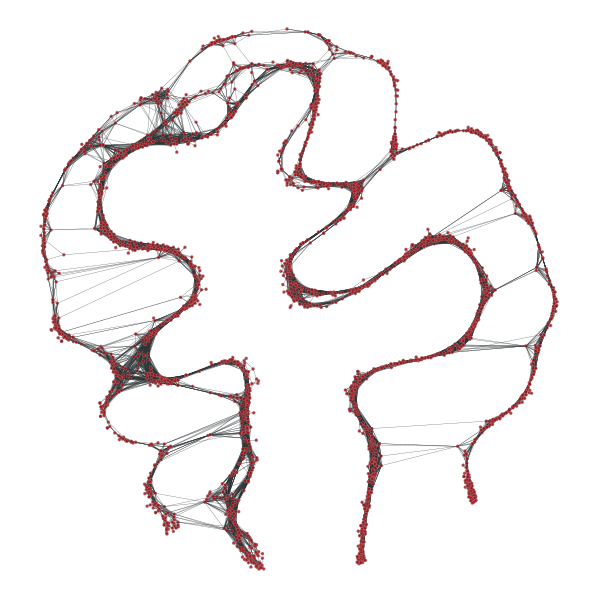} &
    \includegraphics[width=0.33\columnwidth]{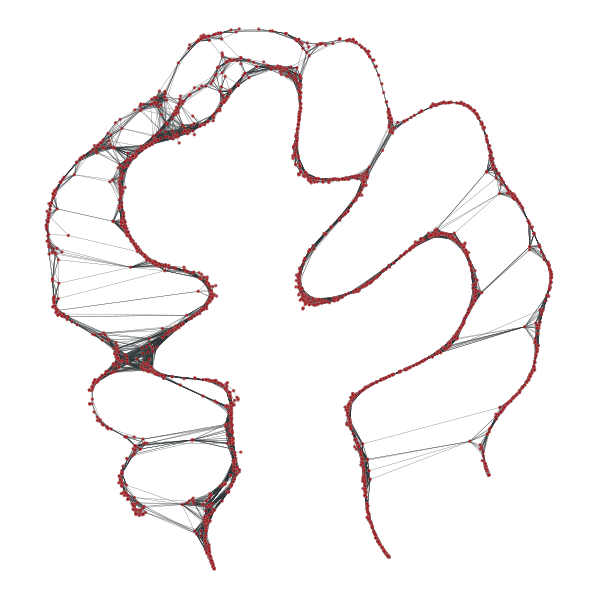} \\ \hline
    \end{tabular}
    \label{tab:viscomp}
\end{table*}

\begin{table*}[h!]
    \centering
    \caption{
    Visual comparison: on graphs GION\_1 and GION\_3, L-tsNET produces less distortion of long paths and cycles compared to the other algorithms, producing better neighborhood preservation and stress.
    }
    \begin{tabular}{|c|c|c|c|}
    \hline
    tsNET & BH-tsNET & FIt-tsNET & L-tsNET \\ \hline
    \multicolumn{4}{|c|}{GION\_1} \\ \hline
    \includegraphics[width=0.33\columnwidth]{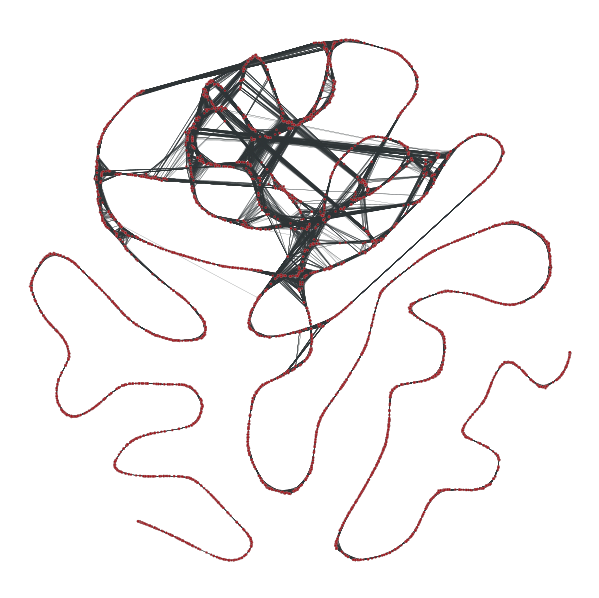} &
    \includegraphics[width=0.33\columnwidth]{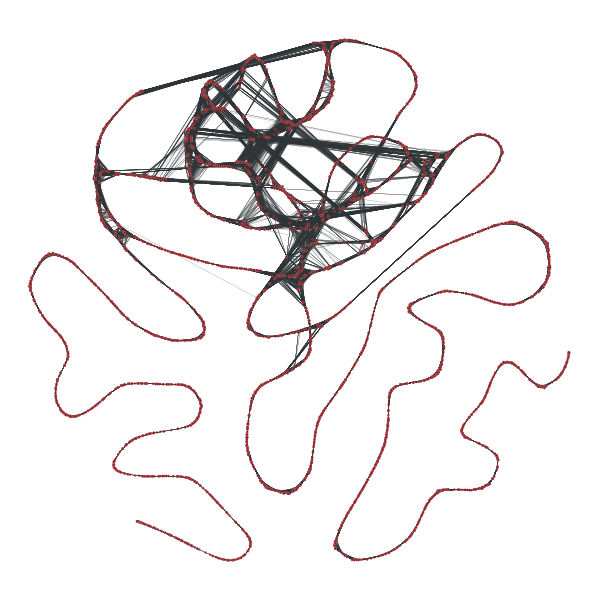} &
    \includegraphics[width=0.33\columnwidth]{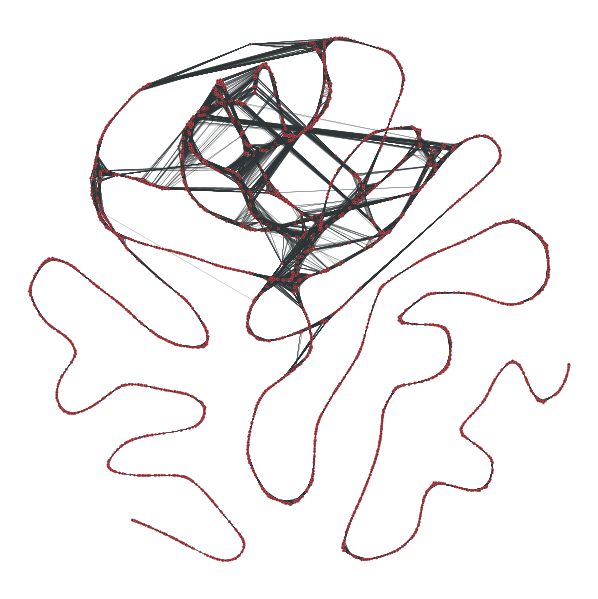} &
    \includegraphics[width=0.33\columnwidth]{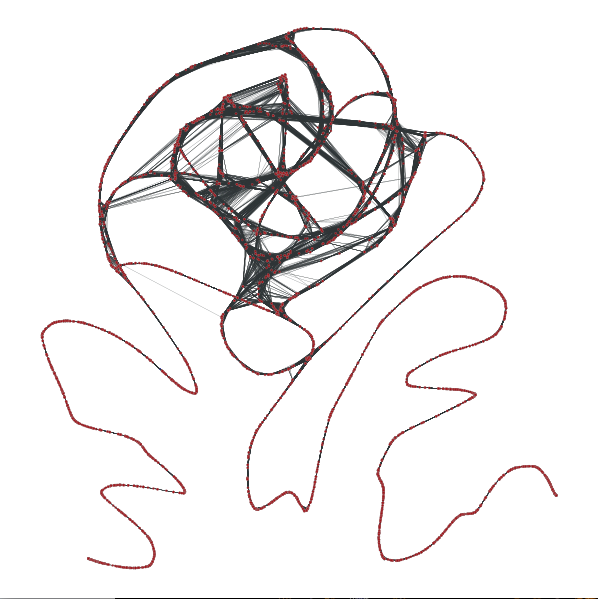}  \\ \hline
    \multicolumn{4}{|c|}{GION\_3} \\ \hline
    \includegraphics[width=0.33\columnwidth]{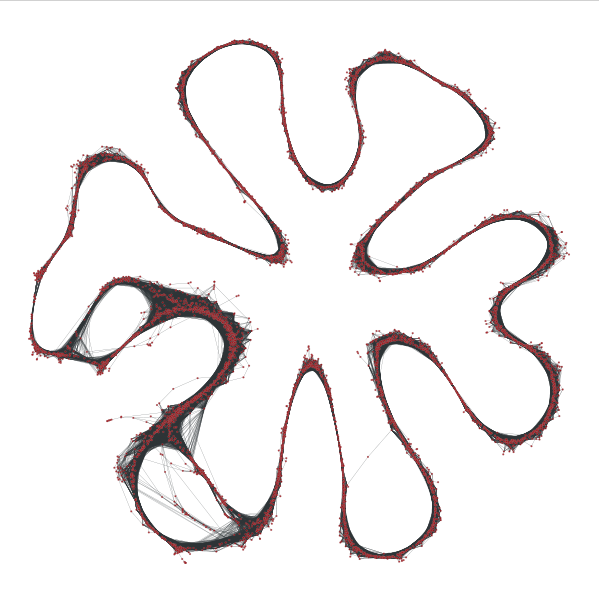} &
    \includegraphics[width=0.33\columnwidth]{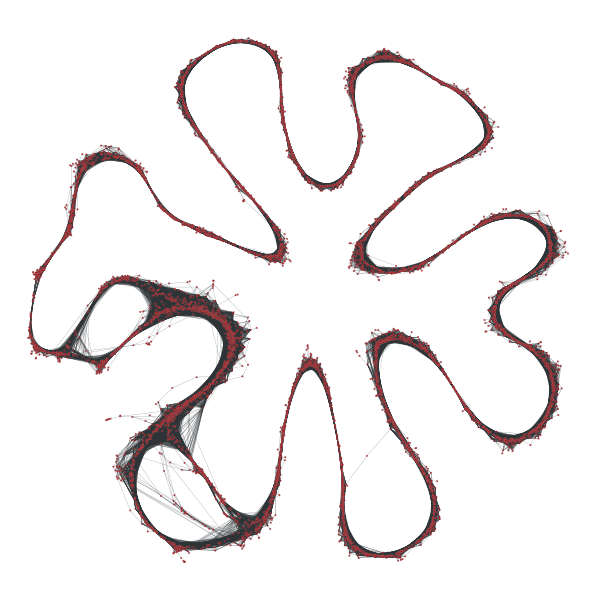} &
    \includegraphics[width=0.33\columnwidth]{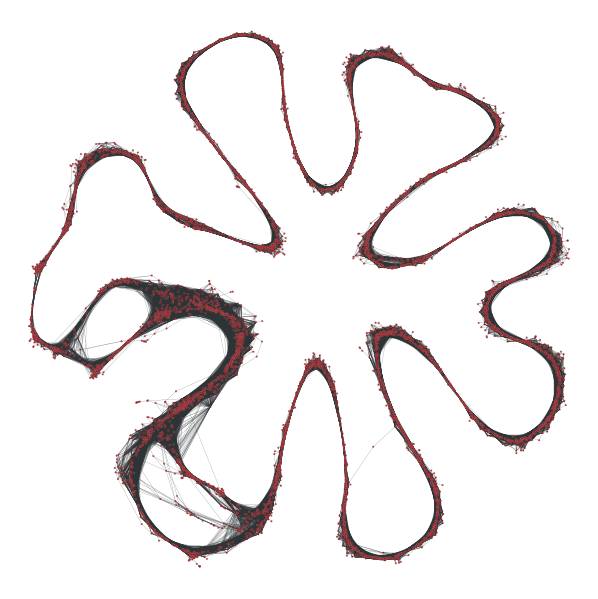} &
    \includegraphics[width=0.33\columnwidth]{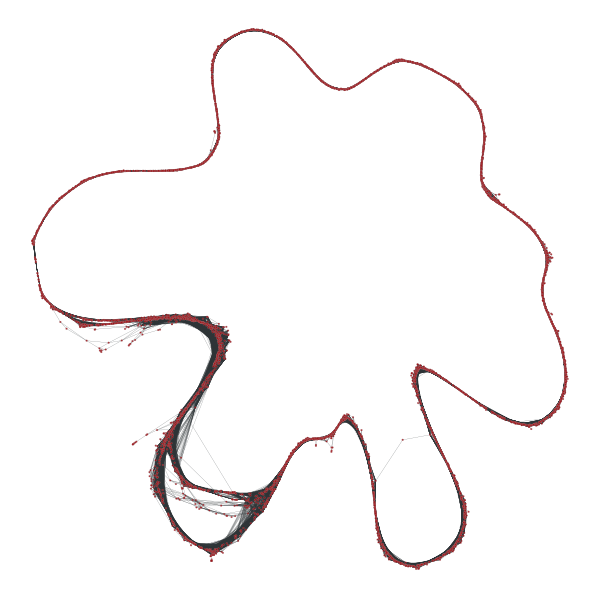}  \\ \hline
    \end{tabular}
    \label{tab:viscomp_good}
\end{table*}

\emph{In summary, \texttt{BH-tsNET}, \texttt{FIt-tsNET}, and \texttt{L-tsNET} obtain about the same quality drawings as \texttt{tsNET} on neighborhood preservation, stress, edge crossing, and 
shape-based metrics with much faster runtime,
confirming Hypothesis \ref{hyp:quality}.}

\begin{table*}[h!]
    \centering
    \caption{
    Visual comparison: For graphs migrations and as19990606, \texttt{BH-tsNET} and \texttt{FIt-tsNET} produce drawings that more clearly show clusters and stars compared to \texttt{tsNET}, untangling the structure rather than overly compacting them.
    }
    \begin{tabular}{|c|c|c|c|}
    \hline
    tsNET & BH-tsNET & FIt-tsNET & L-tsNET \\ \hline
    \multicolumn{4}{|c|}{migrations} \\ \hline
    \includegraphics[width=0.33\columnwidth]{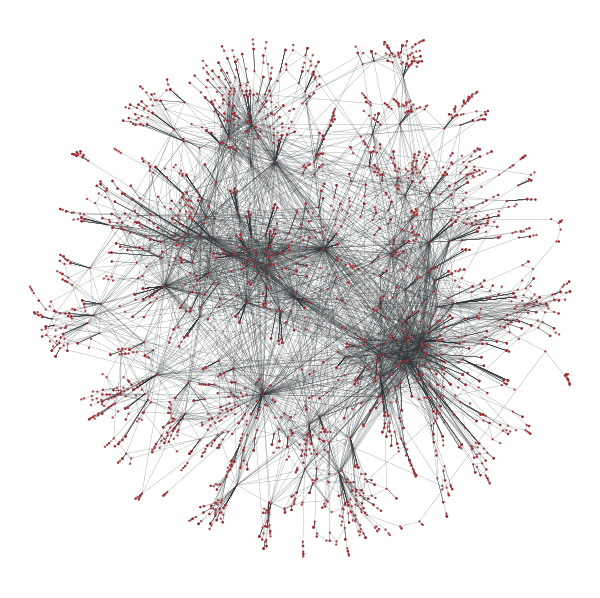} &
    \includegraphics[width=0.33\columnwidth]{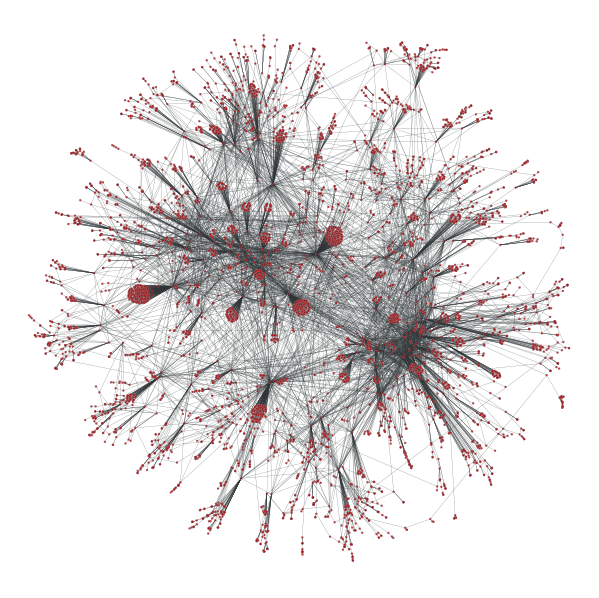} &
    \includegraphics[width=0.33\columnwidth]{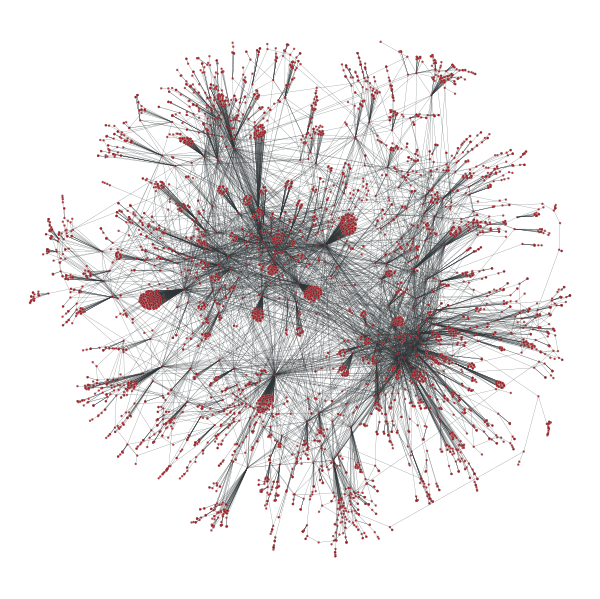} &
    \includegraphics[width=0.33\columnwidth]{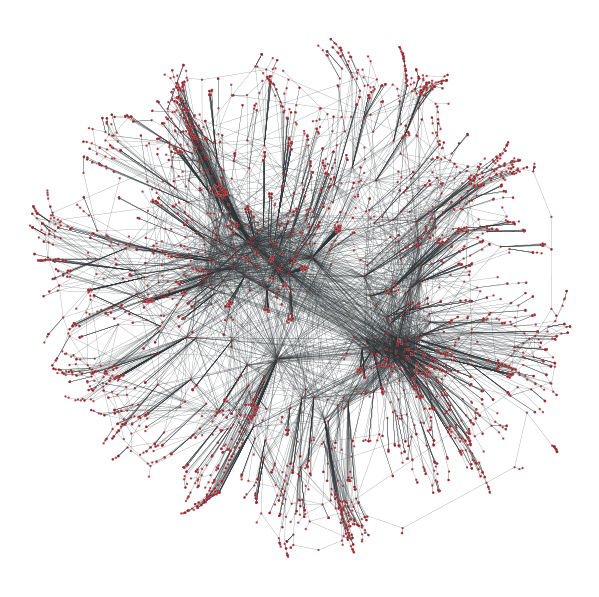} \\ \hline
    \multicolumn{4}{|c|}{as19990606} \\ \hline
    \includegraphics[width=0.33\columnwidth]{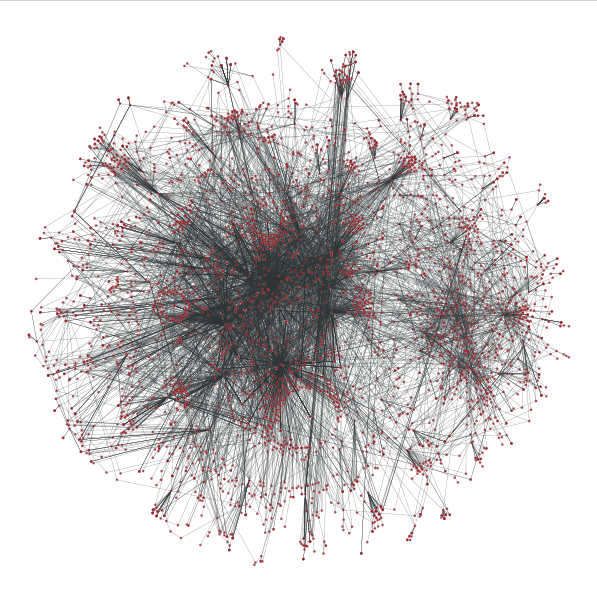} &
    \includegraphics[width=0.33\columnwidth]{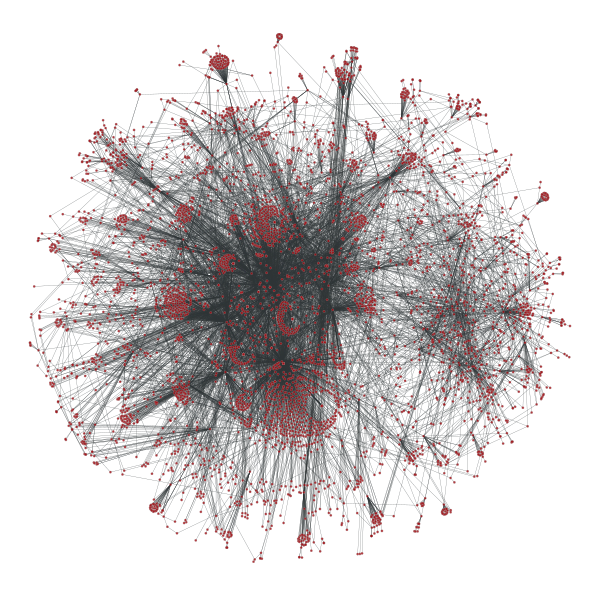} &
    \includegraphics[width=0.33\columnwidth]{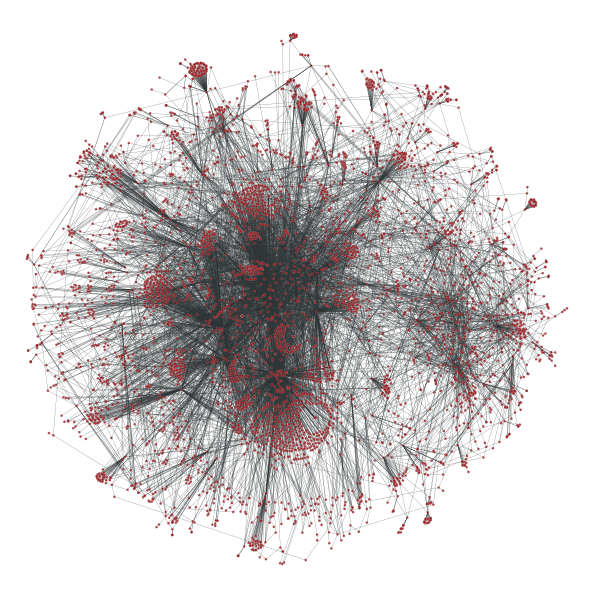} &
    \includegraphics[width=0.33\columnwidth]{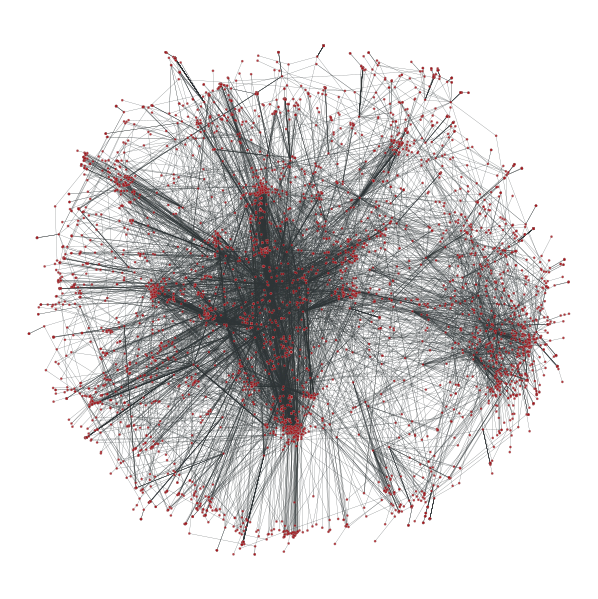} \\ \hline
    \end{tabular}
    \label{tab:viscomp_diff}
\end{table*}

\subsection{Visual Comparison}

In addition to quantitative evaluation with quality metrics, 
We also qualitatively compare the quality of drawings produced by each algorithm with visual comparison.

For visual comparison, \texttt{BH-tsNET}, \texttt{FIt-tsNET}, and \texttt{L-tsNET} produce almost the same drawings as \texttt{tsNET}, maintaining the same visual quality with much faster runtime. 
Table \ref{tab:viscomp} shows detailed visual comparisons of drawings 
on three graphs (mesh graph 3elt, scale-free graph oflights, GION graph GION\_7), where all algorithms produce visually almost identical drawings. 

Surprisingly, \texttt{L-tsNET} produces better quality drawings than the other three algorithms for some larger graphs with long diameters.
Table \ref{tab:viscomp_good} 
shows examples where \texttt{L-tsNET} produce drawings with better quality metrics, most notably neighborhood preservation, than the other three algorithms. 
For GION\_1, the drawing by \texttt{L-tsNET} has fewer bends on the long paths, leading to better neighborhood preservation and less stress. 
The improvement is even more stark for GION\_3, with the long cycle drawn even smoother with much fewer excessive bends than the other three algorithms, resulting in notable improvement in neighborhood preservation and shape-based metrics.

Furthermore, for some scale-free graphs, \texttt{BH-tsNET} and \texttt{FIt-tsNET} produce drawings untangling locally dense clusters/stars,  revealing the intra-cluster structure more clearly than \texttt{tsNET}. 
For example,  Table 
\ref{tab:viscomp_diff} 
shows the drawings of the migrations and as19990606 graphs, where several clusters and stars drawn too compactly by \texttt{tsNET} are nicely untangled by \texttt{BH-tsNET} and \texttt{FIt-tsNET}. 
Although this leads to a lower shape-based metric due to degree one vertices being drawn closer to each other than the high-degree hubs they are actually connected to, the untangling reveals star structures more clearly in \texttt{BH-tsNET} and \texttt{FIt-tsNET} than \texttt{tsNET}.

\emph{In summary, visual comparison shows that \texttt{BH-tsNET}, \texttt{FIt-tsNET}, and \texttt{L-tsNET} compute almost identical drawings to \texttt{tsNET},
confirming Hypothesis \ref{hyp:quality}.
Moreover, for large graphs with long diameters, \texttt{L-tsNET} produces better drawings with much fewer bends in cycles and paths, 
than the other algorithms, while for scale-free graphs with locally dense clusters, \texttt{BH-tsNET} and \texttt{FIt-tsNET} produce better drawings, more clearly showing the intra-cluster structures than \texttt{tsNET}.
}

\subsection{Summary and Discussion}

In summary, our extensive experiments support Hypotheses \ref{hyp:runtime_base}, \ref{hyp:runtime_comp}, and \ref{hyp:runtime_lts}: \texttt{BH-tsNET}, \texttt{FIt-tsNET}, and \texttt{L-tsNET} run significantly faster than \texttt{tsNET}, \texttt{FIt-tsNET} runs  faster than \texttt{BH-tsNET}, and  \texttt{L-tsNET} is the fastest. 
Surprisingly, our algorithms compute almost the same quality drawings as \texttt{tsNET}, in terms of quality metrics and visual comparison, 
supporting  Hypothesis \ref{hyp:quality},
as seen in Figure \ref{fig:metrics_drg_avg_all}.

With the overall runtime of \texttt{FIt-tsNET} being $O(n \log n)$ and \texttt{L-tsNET} being $O(n)$, it might be expected that \texttt{L-tsNET} would have an even larger runtime improvement over \texttt{FIt-tsNET} than the observed 66.3\% improvement. 
However, it should be noted that the cost minimization iteration of \texttt{FIt-tsNET} is reduced to $O(n)$ time by the FFT-accelerated interpolation, which brings the runtime of \texttt{FIt-tsNET} closer to \texttt{L-tsNET}. 
Regardless, \texttt{L-tsNET} still runs significantly faster with similar quality metrics, demonstrating its strengths.
The results of the comparison experiments are  summarized as follows:

\begin{itemize}
    \item On average, \texttt{BH-tsNET}, \texttt{FIt-tsNET}, and \texttt{L-tsNET} run 93.5\%, 96\%, and 98.6\% faster than \texttt{tsNET}. 
    \item Overall, \texttt{BH-tsNET}, \texttt{FIt-tsNET}, and \texttt{L-tsNET}  compute almost the same quality drawings as \texttt{tsNET} on neighborhood preservation, stress, shape-based metrics, and edge crossing metrics.
    
    \item Visual comparisons confirm that our algorithms generally compute almost the same quality drawings as \texttt{tsNET}, consistent with quality metrics comparison.
\end{itemize}

While in general the drawings produced by all four algorithms are very similar, \texttt{L-tsNET} produce notably different drawings than the other algorithms for graphs with certain properties. 
For example, on some large graphs with long diameters, \texttt{L-tsNET} produce drawings with much fewer bends on the long paths and cycles, resulting in better quality metrics for neighborhood preservation, stress, and edge crossing than the other algorithms (see Table \ref{tab:viscomp_good}).

This may be due to the interpolation for entropy. 
In \texttt{tsNET}, entropy pushes the vertices to spread out more evenly throughout the drawing space, which may cause long paths and cycles to coil up to fill the space. 
With \texttt{L-tsNET}, the fewer entropy computations due to the interpolation may weaken this effect, removing the excessive bends.

In other cases, \texttt{L-tsNET} compute drawings very similar to \texttt{tsNET} while \texttt{BH-tsNET} and \texttt{FIt-tsNET} compute drawings that are notably different to \texttt{tsNET}. 
For example, on some scale-free graphs with lower density,  
\texttt{BH-tsNET} and \texttt{FIt-tsNET} compute drawings with dense clusters and large stars being untangled better than \texttt{tsNET} (see Table \ref{tab:viscomp_diff}). 

This may be due to the use of quadtree for KL divergence gradient,
since only the interactions between geometrically close vertices are computed exactly, while those between distant vertices are only approximated.
Namely, the ``repulsion'' between nearby vertices in dense clusters may become comparatively larger than the approximated repulsion from far away vertices, resulting in the vertices in the dense area being spread out. While the drawings by \texttt{L-tsNET} do not share this property, the scores on all metrics are still very similar to \texttt{tsNET}, ensuring a good level of quality is still preserved overall.

Another surprising finding is that \texttt{L-tsNET} obtains lower edge crossings than \texttt{BH-tsNET} and \texttt{FIt-tsNET}, and almost the same as \texttt{tsNET}. 
This may show trade-offs between using quadtree-based entropy and interpolation-based entropy, where the former performs better for showing the intra-cluster structures of dense clusters, while the latter is better for drawings with fewer edge crossings.

\begin{table*}[h!]
    \centering
    \caption{Visual comparison of \texttt{DRGraph} to \texttt{tsNET}, \texttt{BH-tsNET}, \texttt{FIt-tsNET}, and \texttt{L-tsNET}. On scale-free graphs, our algorithms tend to spread out the vertices more evenly overall, while \texttt{DRGraph} tends to produce a ``blob'' with some very long edges jutting out. On mesh graphs, our algorithms tend to preserve the regular structure of the mesh and untangle the structure better than \texttt{DRGraph}.}
    \begin{tabular}{|c|c|c|c|c|}
    \hline
    tsNET & BH-tsNET & FIt-tsNET & L-tsNET & DRGraph \\ \hline
    \multicolumn{5}{|c|}{CA-GrQc} \\ \hline
    \includegraphics[width=0.33\columnwidth]{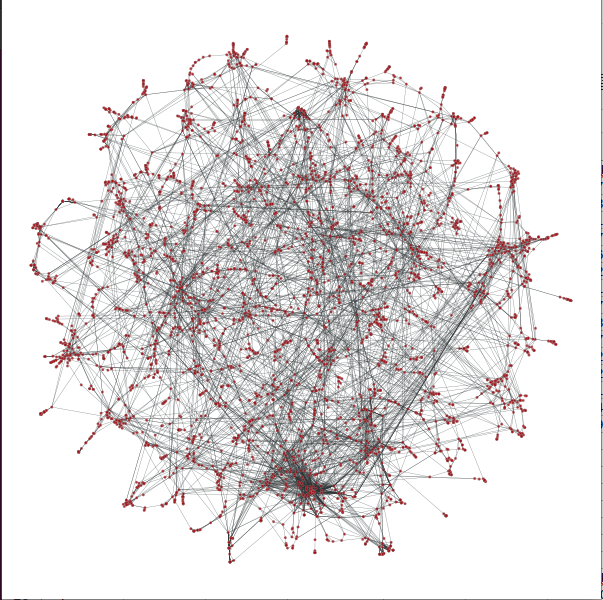} &
    \includegraphics[width=0.33\columnwidth]{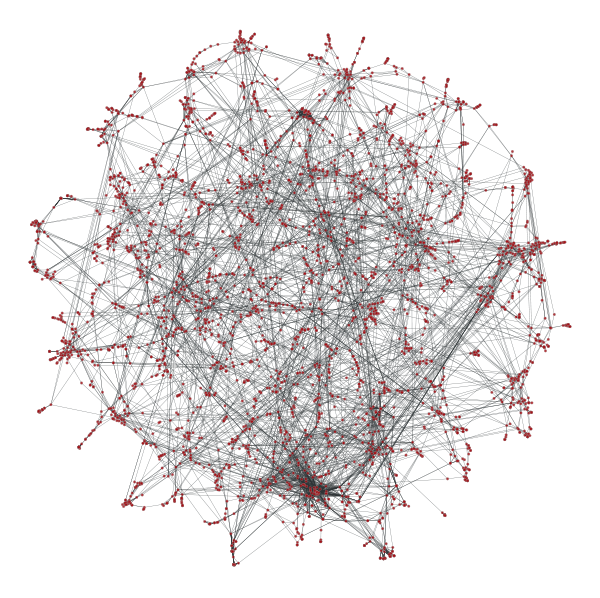} &
    \includegraphics[width=0.33\columnwidth]{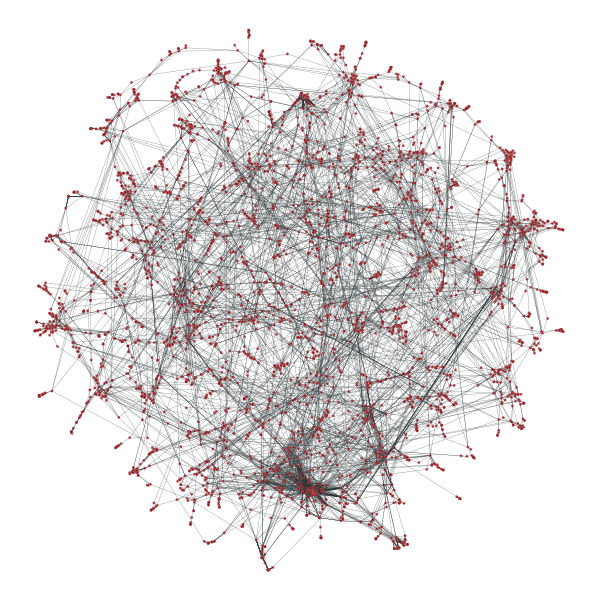} &
    \includegraphics[width=0.33\columnwidth]{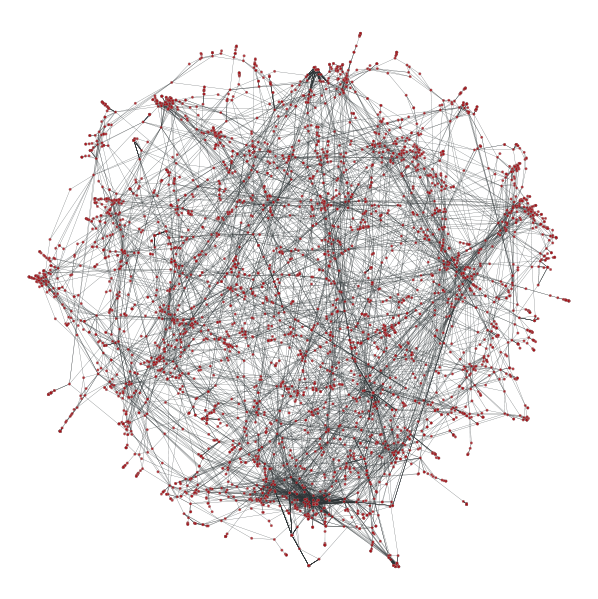} &
    \includegraphics[width=0.33\columnwidth]{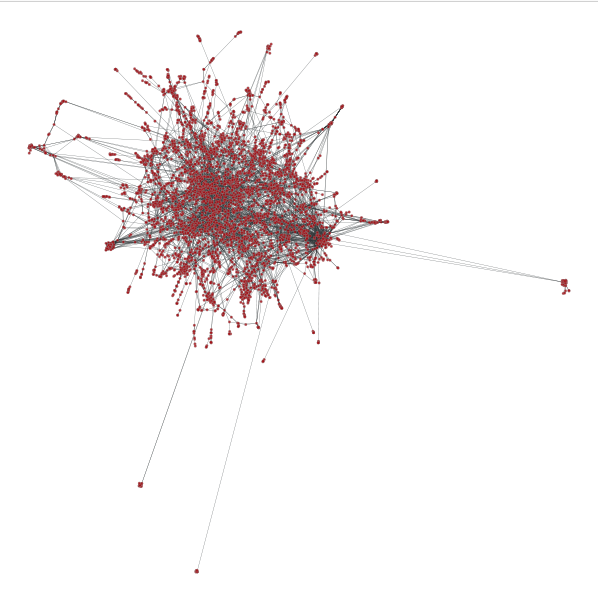} \\ \hline
    \multicolumn{5}{|c|}{lastfm\_asia} \\ \hline
    \includegraphics[width=0.33\columnwidth]{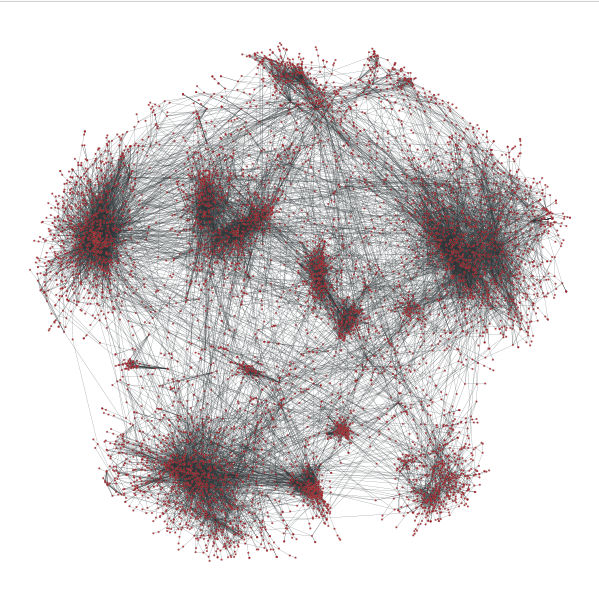} &
    \includegraphics[width=0.33\columnwidth]{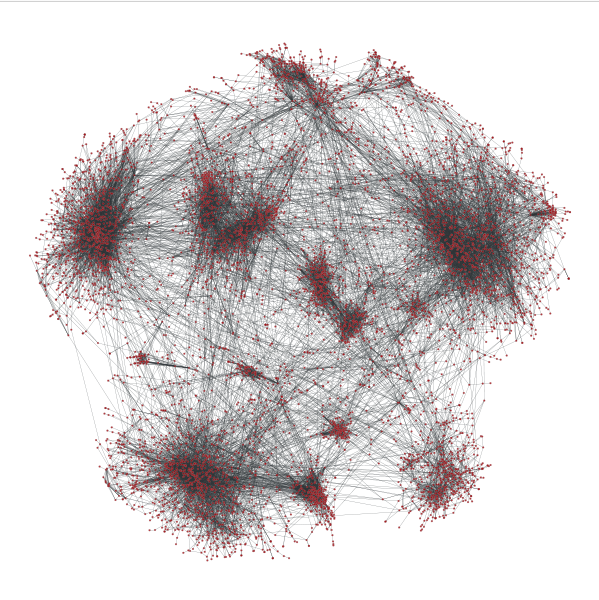} &
    \includegraphics[width=0.33\columnwidth]{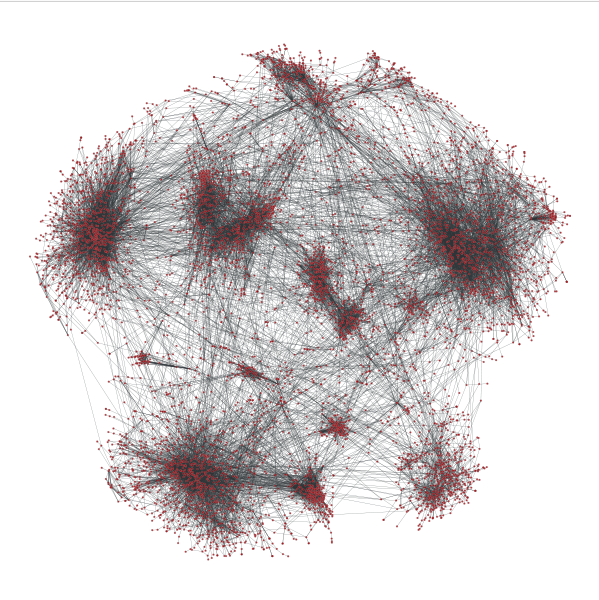} &
    \includegraphics[width=0.33\columnwidth]{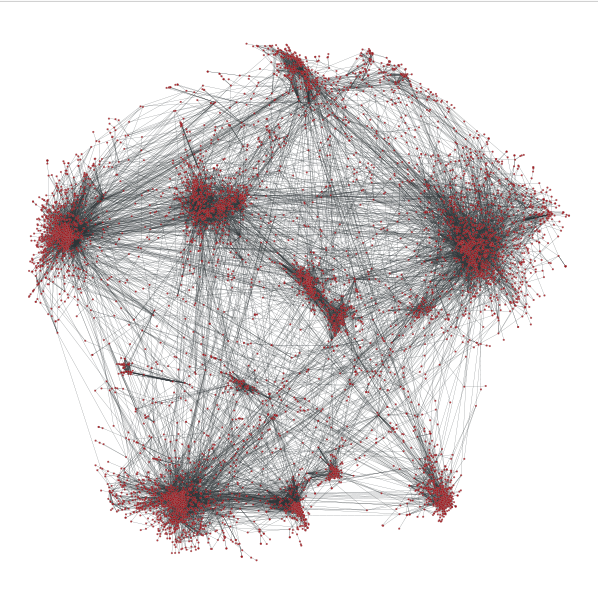} &
    \includegraphics[width=0.33\columnwidth]{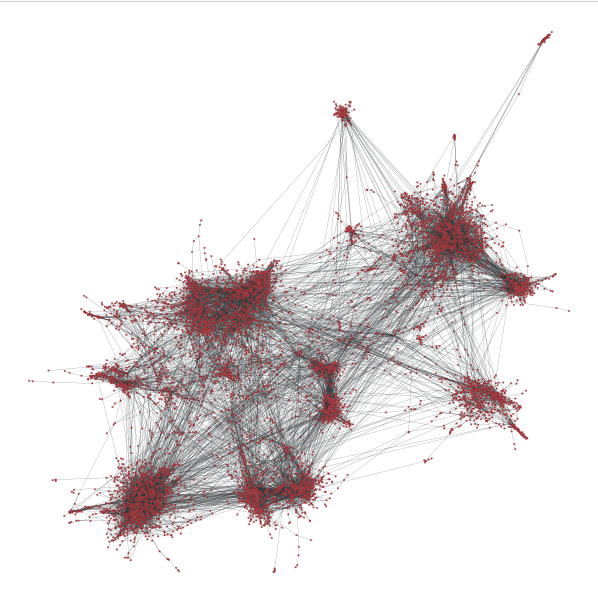} \\ \hline
    \multicolumn{5}{|c|}{bcsstk09} \\ \hline
    \includegraphics[width=0.33\columnwidth]{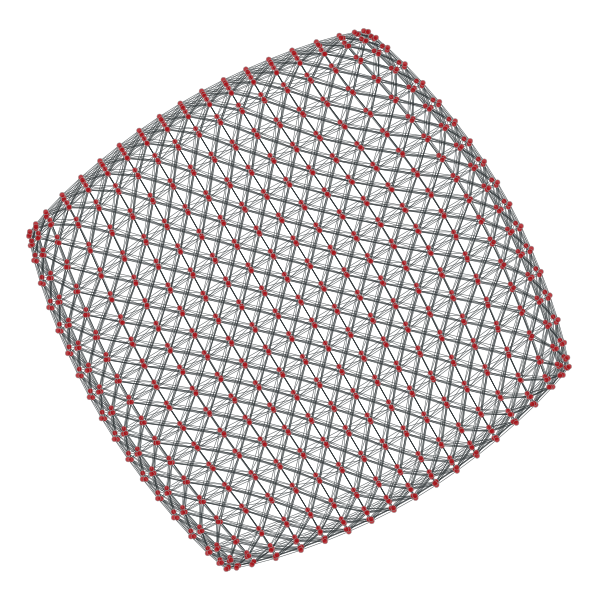} &
    \includegraphics[width=0.33\columnwidth]{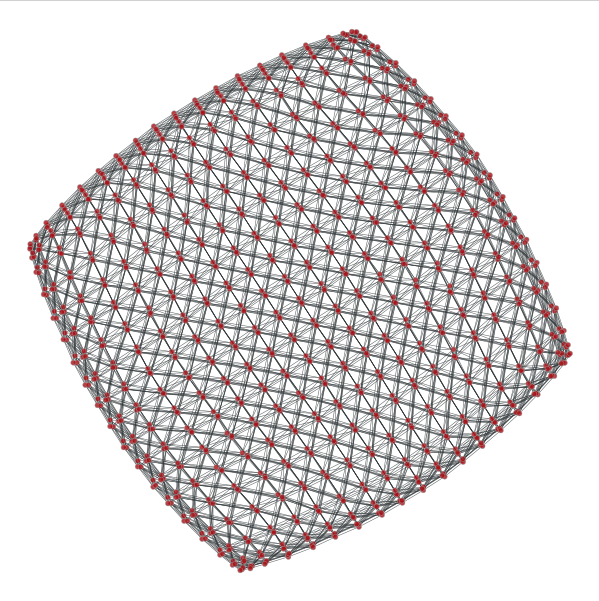} &
    \includegraphics[width=0.33\columnwidth]{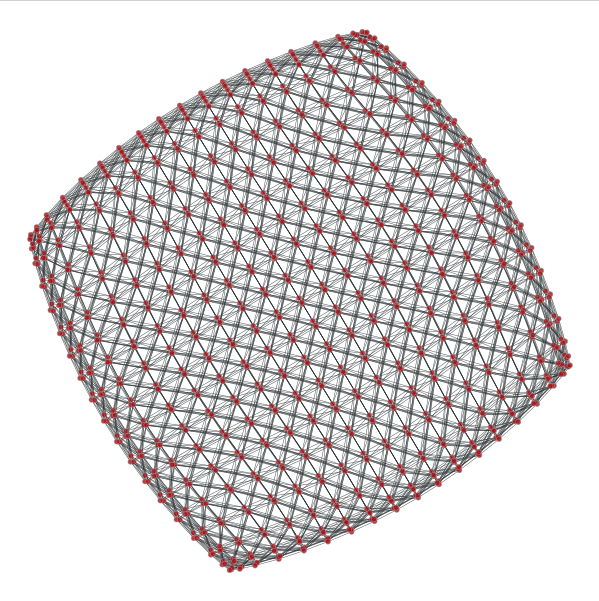}&
    \includegraphics[width=0.33\columnwidth]{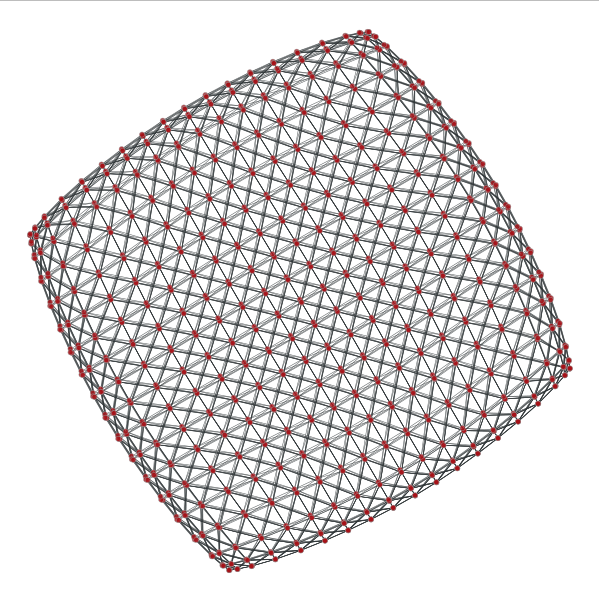} &
    \includegraphics[width=0.33\columnwidth]{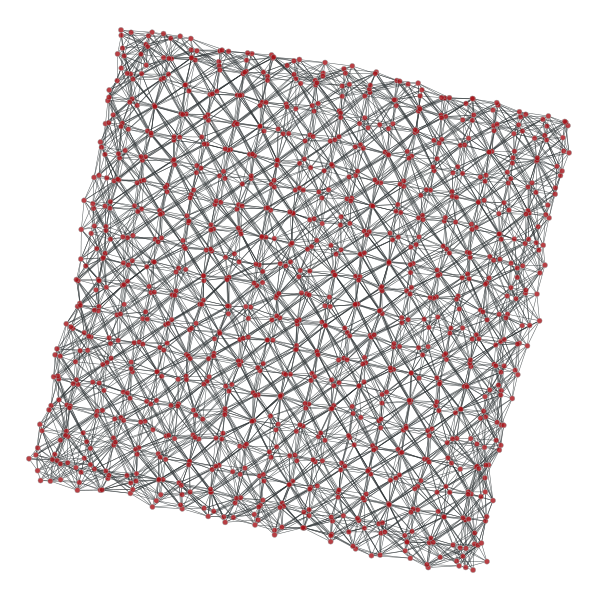} \\ \hline
    \multicolumn{5}{|c|}{data} \\ \hline
    \includegraphics[width=0.33\columnwidth]{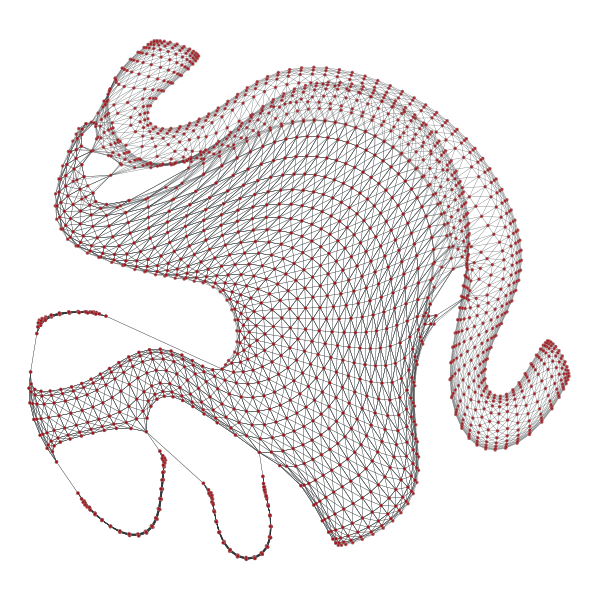} &
    \includegraphics[width=0.33\columnwidth]{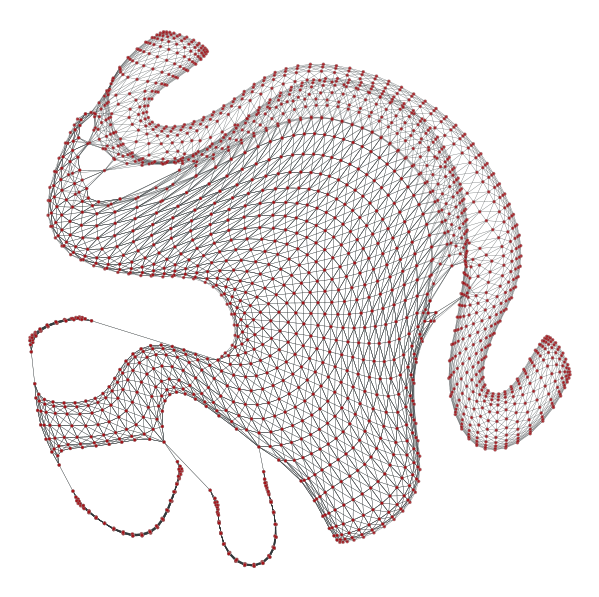} &
    \includegraphics[width=0.33\columnwidth]{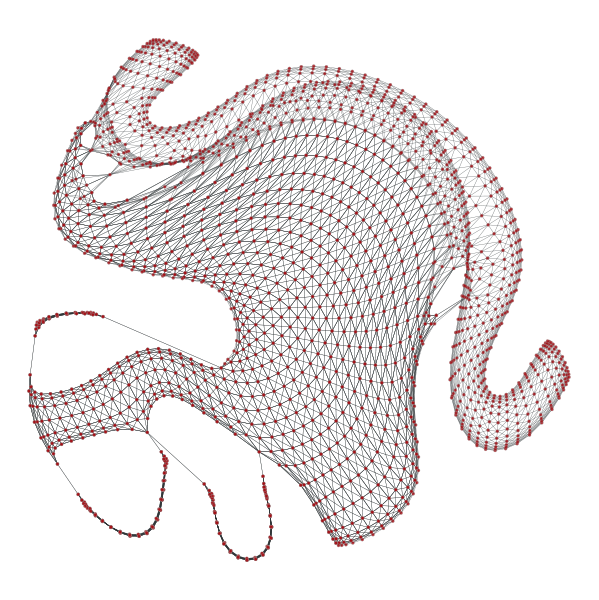} &
    \includegraphics[width=0.33\columnwidth]{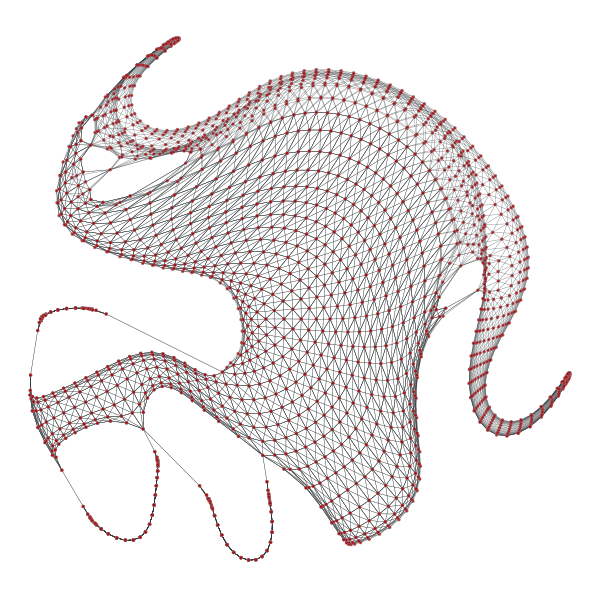} &
    \includegraphics[width=0.33\columnwidth]{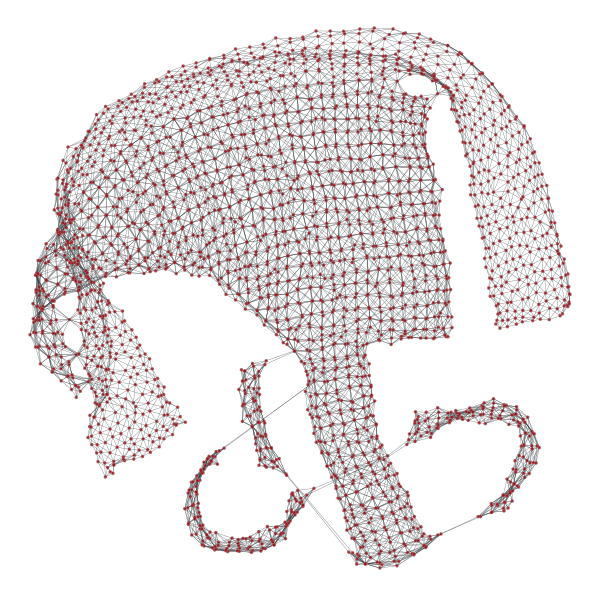} \\ \hline
    \end{tabular}
    \label{tab:viscomp_drg}
\end{table*}

\section{Experiment 2: Comparison with DRGraph}
\label{sec:drg}

In this section, we present a comparison experiment of our fast \texttt{tsNET} algorithms to \texttt{DRGraph}~\cite{zhu2020drgraph},  
a \texttt{t-SNE}-based graph drawing algorithm runs in $O(n + m)$ time using negative sampling, to compare the strengths of each algorithm on runtime, quality metrics, and visual comparison
While \texttt{DRGraph} also uses  \texttt{t-SNE}, its objective function only optimizes the KL divergence without the other two terms (compression and entropy) used by the objective function of \texttt{tsNET}.

We focus on the comparison of our algorithms to \texttt{DRGraph} since it is the most related  (i.e., based on \texttt{t-SNE}), unlike other graph drawing algorithms with vastly different optimization criteria, such as force-directed or stress-based algorithms.

For example, recent fast graph drawing algorithms are based on the {\em force-directed} algorithms or {\em stress minimization} algorithms.
The SubLinearForce~\cite{meidiana2024sublinearforce} is a framework for sublinear-time force-directed algorithms using various sampling methods.
The SSM and SSGD algorithms~\cite{meidiana2021stress} are sublinear-time stress minimization algorithms, using pivots and sampling.
The t-FDP algorithm~\cite{zhong2023force} is a linear-time force-directed algorithm using the t-distribution, where the computation of the force is accelerated using FFT-based interpolation.
However, the objective function of t-FDP is vastly different from \texttt{t-SNE}-based algorithms (i.e., computing the KL divergence between high- and low-dimensional similarities), based on the force-directed model with modified repulsion and attraction forces.

Note that \texttt{DRGraph} was already compared to \texttt{tsNET}, running significantly faster than \texttt{tsNET}, while computing drawings with similar quality~\cite{zhu2020drgraph}. 
However, the comparison was done against ``base'' \texttt{tsNET} (i.e., with random initialization) rather than $tsNET*$ (i.e., with PMDS initialization), which was shown to have a better performance.
Therefore, in our experiments, we compare our algorithms \texttt{BH-tsNET}, \texttt{FIt-tsNET}, and \texttt{L-tsNET} with PMDS initialization to \texttt{DRGraph} using the default parameters suggested in~\cite{zhu2020drgraph}.

\subsection{Runtime  Comparison} 

Figure \ref{fig:metrics_drg_avg_all}(a) shows the average runtime of our algorithms and \texttt{DRGraph}, where  
\texttt{DRGraph} runs slightly faster than \texttt{L-tsNET} on average.

Note that \texttt{DRGraph} is implemented fully in C++, unlike our algorithms implemented partly in Python. 
Therefore, the difference in language efficiency may have factored into the slightly faster runtime of \texttt{DRGraph}.

\subsection{Quality Metrics Comparison}

Figures \ref{fig:metrics_drg_avg_all}(b)-(e) show the average quality metrics comparison between our algorithms and \texttt{DRGraph},
where overall, our algorithms and \texttt{DRGraph} show similar performance on quality metrics.

For example, when we compare \texttt{L-tsNET} and \texttt{DRGraph},  
\texttt{L-tsNET} obtains slightly higher neighborhood preservation, while \texttt{DRGraph} obtains slightly higher shape-based metrics, on average. 
Meanwhile, both algorithms show similar performance on stress and edge crossing metrics, on average.

All our \texttt{BH-tsNET}, \texttt{FIt-tsNET}, and \texttt{L-tsNET} algorithms perform significantly better than \texttt{DRGraph} on the edge crossing metrics for mesh graphs, see Figure \ref{fig:metrics_drg_avg_all}(f). 
In particular, \texttt{DRGraph} computes drawings with 23.5\% higher edge crossing metrics than \texttt{L-tsNET} on mesh graphs. 
Moreover, \texttt{L-tsNET} obtains slightly higher shape-based metrics than \texttt{DRGraph} for mesh graphs.

\subsection{Visual Comparison} 

Table \ref{tab:viscomp_drg} shows examples of visual comparisons between our algorithms and \texttt{DRGraph}, where for {\em mesh graphs},  our algorithms notably outperform \texttt{DRGraph}.
Specifically, our algorithms \texttt{BH-tsNET}, \texttt{FIt-tsNET}, and \texttt{L-tsNET}  produce significantly better quality drawings on {\em mesh graphs} than \texttt{DRGraph} with significantly fewer edge crossings, successfully untangling and showing the regular mesh structure of the graph. 

For example, on bcsstk09, even after setting the parameter $b=3$ as suggested for grid graphs~\cite{zhu2020drgraph}, \texttt{DRGraph} fails to display the regular mesh structure to the same extent as our algorithms, 
with ``zig-zags'' on the grid structure leading to higher edge crossings and lower neighborhood preservation. 
For the data graph, even on the best choice of $b$, \texttt{DRGraph} also ``folds'' the mesh, 
leading to higher edge crossings, 
while our algorithms successfully unfold the mesh structure.

For {\em scale-free graphs}, such as CA-GrQc and lastfm\_asia, our algorithms and \texttt{DRGraph} compute drawings with different characteristics, emphasizing different strengths. 
Our algorithms tend to spread the vertices evenly throughout the drawing area, 
improving neighborhood preservation for these particular graphs
, while \texttt{DRGraph} tends to have vertices in the dense areas (i.e., ``hairballs'') closer together, 
improving stress in the dense areas of drawings. 
For {\em gion graphs}, \texttt{L-tsNET} and \texttt{DRGraph} show similar performance on average.

\subsection{Summary and Discussion}

Our experiments show different strengths for \texttt{L-tsNET} and \texttt{DRGraph}, where \texttt{DRGraph} runs slightly faster, however with worse quality on some graph classes. 
Therefore, we can summarize as follows:

\begin{itemize}
    \item 
    Both \texttt{L-tsNET} and \texttt{DRGraph} have the same linear time complexity. Due to C++ implementation, \texttt{DRGraph} runs slightly faster than \texttt{L-tsNET} implemented in Python.
    \item \texttt{L-tsNET} computes better quality drawings than \texttt{DRGraph} with comparable runtime, especially for graphs with mesh/regular structures.
\end{itemize}

The better performance of \texttt{L-tsNET} compared to \texttt{DRGraph} on mesh graphs, as well as the tendency to spread out ``hairball'' structures of scale-free graphs, may be due to the entropy term included in our algorithms, which is not included in \texttt{DRGraph}. 
Since entropy pushes points in the geometric space to be evenly spaced~\cite{gansner2012maxent}, it may contribute to producing drawings that display the regular structure of the mesh better.

\section{Conclusion}
\label{sec:conclusion}

We close the gap in the literature by significantly improving the time complexity of \texttt{tsNET} algorithm for graph drawing from $O(nm)$-time to $O(n \log n)$-time with \texttt{BH-tsNET} and  \texttt{FIt-tsNET} algorithms, and finally to $O(n)$-time with  \texttt{L-tsNET} algorithm. 

We formally prove the {\em theoretical} time complexity of our new algorithms to guarantee scalability.
Moreover, extensive evaluation with real-world large and complex benchmark data sets demonstrates that \texttt{BH-tsNET}, \texttt{FIt-tsNET}, and \texttt{L-tsNET} achieve significant runtime improvements over \texttt{tsNET} (93.5\%, 96\%, 98.6\%), while preserving almost the same quality on neighborhood preservation, stress, edge crossings and shape-based metrics, as well as visual comparison.

We also compare our algorithms to \texttt{DRGraph}~\cite{zhu2020drgraph}, a linear-time \texttt{t-SNE}-based graph drawing algorithm with a different optimization function. 
Experiments show that our algorithms perform especially well over \texttt{DRGraph} for graphs with regular structures, such as mesh graphs.

Future work includes improving the {\em theoretical} runtime to sublinear-time. 
Recent sublinear-time force-directed~\cite{meidiana2020sublinear,meidiana2024sublinearforce} and sublinear-time stress minimization algorithms~\cite{meidiana2021stress} shed light on this direction. 

Another direction is to improve {\em practical} runtime using parallelism such as GPU, for example utilizing \texttt{t-SNE-CUDA}~\cite{chan2018tsne} and \texttt{GPGPU-SNE}~\cite{pezzotti2020gpgpu}, and conduct experiments with extremely big and complex graphs.


\bibliographystyle{abbrv-doi-hyperref}
\bibliography{template}


\section{Appendix}

\begin{table}[h]
    \centering
    \caption{Data sets for experiments} 
    \vspace{-2mm}
    \subfloat[Benchmark graphs]{
    \begin{tabular}{|l|c|c|c|c|}
    \hline
        graph & $|V|$ & $|E|$ & $u$ & density \\ \hline
        G\_13\_0 & 1647 & 6487 & 300 & 3.94 \\ \hline
        soc\_h & 2000 & 16097 & 400 & 8.05 \\ \hline
        Block\_2000 & 2000 & 9912 & 40 & 4.96 \\ \hline
        G\_4\_0 & 2075 & 4769 & 100 & 2.30 \\ \hline
        oflights & 2905 & 15645 & 300 & 5.39 \\ \hline
        tvcg & 3213 & 10140 & 100 & 3.16 \\ \hline
        facebook & 4039 & 88234 & 1100 & 21.8 \\ \hline
        CA-GrQc & 4158 & 13422 & 40 & 3.23 \\ \hline
        EVA & 4475 & 4652 & 600 & 1.04 \\ \hline
        us\_powergrid & 4941 & 6594 & 40 & 1.33 \\ \hline
        as19990606 & 5188 & 9930 & 1200 & 1.91 \\ \hline
        migrations & 6025 & 9378 & 600 & 1.57 \\ \hline
        lastfm\_asia & 7624 & 27906 & 100 & 3.65\\ \hline
        CA-HepTh & 8638 & 24827 & 100 & 2.87 \\ \hline
        CA-HepPh & 11204 & 117649 & 500 & 10.5 \\ \hline
    \end{tabular}
    }
    \qquad
    \vspace{-2mm}
    \subfloat[GION graphs]{
    \begin{tabular}{|l|c|c|c|c|}
    \hline
        graph & $|V|$ & $|E|$ & $u$ & density \\ \hline
        GION\_2 & 1159 & 6424 & 40 & 5.54 \\ \hline
        GION\_5 & 1748 & 13957 & 100 & 7.98 \\ \hline
        GION\_6 & 1785 & 20459 & 100 & 11.5 \\ \hline
        GION\_7 & 3010 & 41757 & 100 & 13.9 \\ \hline
        GION\_8 & 4924 & 52502 & 100 & 10.7 \\ \hline
        GION\_1 & 5452 & 118404 & 200 & 21.7 \\ \hline
        GION\_4 & 5953 & 186279 & 200 & 31.3\\ \hline
        GION\_3 & 7885 & 427406 & 500 & 54.2\\ \hline
    \end{tabular}
    }
    \qquad
    \vspace{-2mm}
    \subfloat[Mesh graphs]{
    \begin{tabular}{|l|c|c|c|c|}
    \hline
        graph & $|V|$ & $|E|$ & $u$ & density \\ \hline
        cage8 & 1015 & 4994 & 40 & 4.92 \\ \hline
        bcsstk09 & 1083 & 8677 & 40 & 8.01 \\ \hline
        nasa1824 & 1824 & 18692 & 40 & 10.2 \\ \hline
        plat1919 & 1919 & 15240 & 40 & 7.94 \\ \hline
        sierpinski3d & 2050 & 6144 & 40 & 3.00 \\ \hline
        3elt & 4720 & 13722 & 40 & 2.91 \\ \hline
        crack & 10240 & 30380 & 40 & 2.97 \\ \hline
    \end{tabular}
    }
    \label{table:data}
    \vspace{-3mm}
\end{table}


\begin{figure*}[h]
    \centering
    \includegraphics[width=0.9\textwidth]{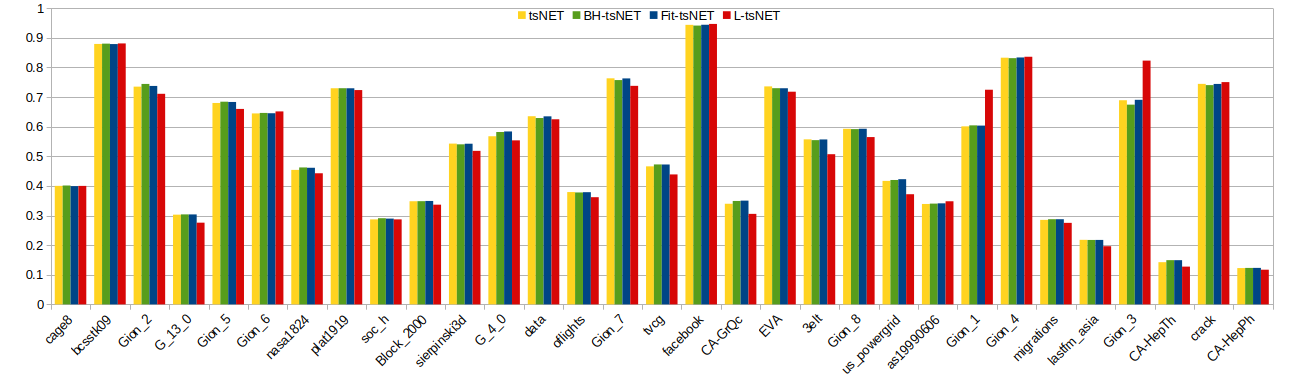}
    \caption{Neighborhood preservation metric (higher = better): On average, the metrics of drawings computed by \texttt{tsNET}, \texttt{BH-tsNET}, \texttt{FIt-tsNET} and \texttt{L-tsNET} are almost the same.}
    \label{fig:tsnet_np}
\end{figure*}


\begin{figure*}[h]
    \centering
    \includegraphics[width=0.9\textwidth]{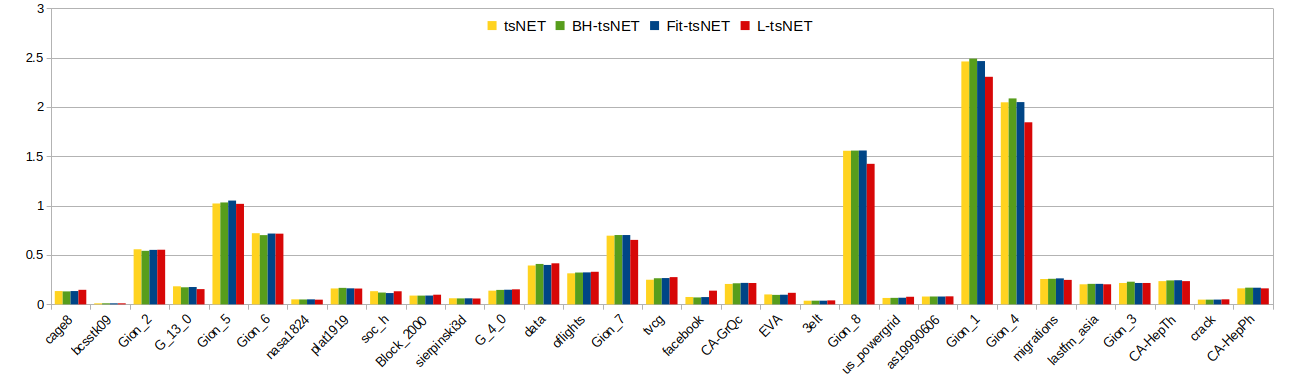}
    \caption{Stress metric (lower = better):  On average, the stress of drawings computed by \texttt{tsNET}, \texttt{BH-tsNET}, and \texttt{FIt-tsNET} 
    are almost the same.}
    \label{fig:tsnet_stress}
\end{figure*}


\begin{figure*}[h]
    \centering
    \includegraphics[width=0.9\textwidth]{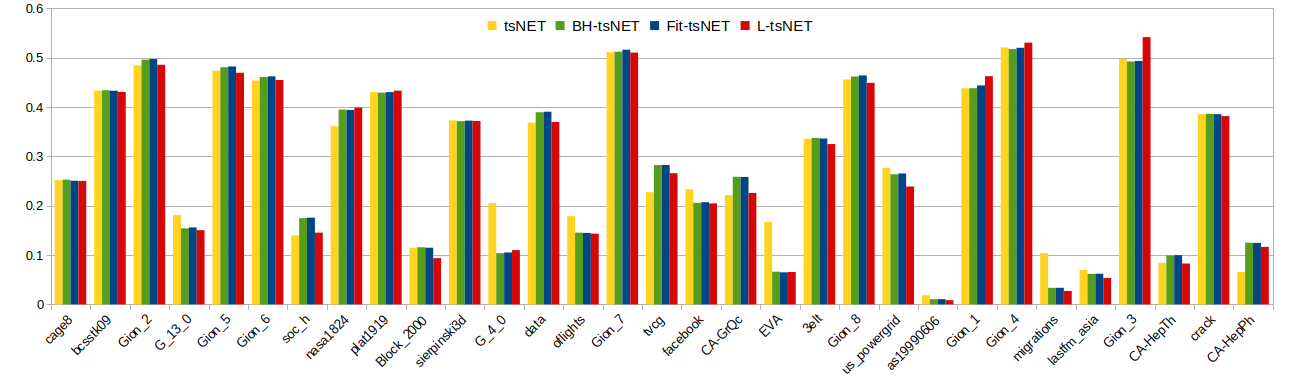}
    \caption{Shape-based metric (higher = better):  On average, the metrics of drawings computed by \texttt{tsNET}, \texttt{BH-tsNET}, \texttt{FIt-tsNET}, and \texttt{L-tsNET} tend to be similar, with \texttt{tsNET} sometimes obtaining better shape-based metrics on sparse (density $<$ 3) graphs and \texttt{BH-tsNET}, \texttt{FIt-tsNET}, and \texttt{L-tsNET} obtaining better shape-based metrics on some denser (density $>$ 3) graphs.}
    \label{fig:tsnet_shp}
\end{figure*}

\begin{figure*}[h]
    \centering
    \includegraphics[width=0.9\textwidth]{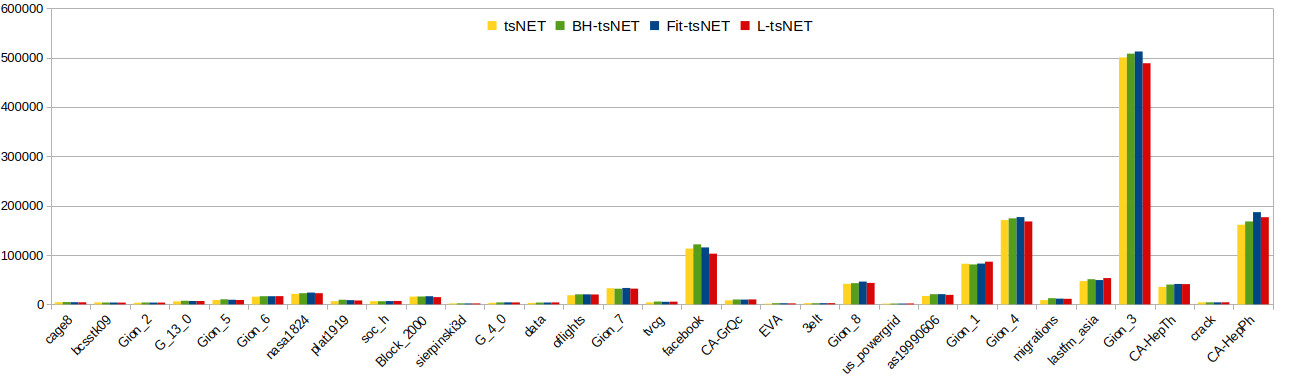}
    \caption{Edge crossing metric (lower = better): On average, the metrics of drawings computed by \texttt{tsNET}, \texttt{BH-tsNET}, \texttt{FIt-tsNET} and \texttt{L-tsNET} tend to be almost the same, with surprisingly lower edge crossing metrics obtained by \texttt{L-tsNET} on graphs with long diameters, e.g. Gion\_4, Gion\_3.}
    \label{fig:tsnet_crossing}
\end{figure*}

\end{document}